\documentclass[onefignum,onetabnum]{siamart171218}

\usepackage{lipsum}
\usepackage{amsfonts}
\usepackage{graphicx}
\usepackage{epstopdf}
\usepackage{algorithmic}
\ifpdf
  \DeclareGraphicsExtensions{.eps,.pdf,.png,.jpg}
\else
  \DeclareGraphicsExtensions{.eps}
\fi

\newsiamremark{remark}{Remark}
\newsiamremark{hypothesis}{Hypothesis}
\crefname{hypothesis}{Hypothesis}{Hypotheses}
\newsiamthm{claim}{Claim}

\headers{Bounded-Confidence Models with Adaptive Confidence Bounds}{G. J. Li, J. Luo, M. A. Porter}

\title{Bounded-Confidence Models of Opinion Dynamics with Adaptive Confidence Bounds\thanks{Submitted to the editors DATE.
\funding{GJL and JL acknowledge funding from NSF grant number 1829071. GJL and MAP acknowledge financial support from the National Science Foundation (grant number 1922952) through the Algorithms for Threat Detection (ATD) program}}}

\author{Grace J. Li\thanks{These two authors contributed equally. \\ Department of Mathematics, University of California, Los Angeles, CA, USA
  (\email{graceli@math.ucla.edu},
  \email{jerryluo8@math.ucla.edu}
  ).}
\and Jiajie Luo\footnotemark[2]
\and Mason A. Porter\thanks{Department of Mathematics, University of California, Los Angeles, CA, USA; Department of Sociology, University of California, Los Angeles, CA, USA; Santa Fe Institute, Santa Fe, NM, USA (\email{mason@math.ucla.edu}).}
}

\usepackage{amsopn}

\makeatletter
\newcommand*{\addFileDependency}[1]{% argument=file name and extension
  \typeout{(#1)}% latexmk will find this if $recorder=0 (however, in that case, it will ignore #1 if it is a .aux or .pdf file etc and it exists! if it doesn't exist, it will appear in the list of dependents regardless)
  \@addtofilelist{#1}% if you want it to appear in \listfiles, not really necessary and latexmk doesn't use this
  \IfFileExists{#1}{}{\typeout{No file #1.}}% latexmk will find this message if #1 doesn't exist (yet)
}
\makeatother

\newtheorem{thm}{Theorem}

\usepackage{cite}

\usepackage{multirow}
\usepackage{tabularx,booktabs}

\usepackage{graphicx,epstopdf} % <- Preamble
\usepackage[caption=false]{subfig} % <- Preamble
\usepackage[T1]{fontenc}

\ifpdf
\hypersetup{
  pdftitle={Bounded-Confidence Models of Opinion Dynamics with Adaptive Confidence Bounds},
  pdfauthor={G. J. Li, J. Luo, and M. A. Porter}
}
\fi

\begin{document}

\maketitle

% REQUIRED
\begin{abstract}
People’s opinions change with time as they interact with each other. In a bounded-confidence model (BCM) of opinion dynamics, individuals (which are represented by the nodes of a network) have continuous-valued opinions and are influenced by neighboring nodes whose opinions are sufficiently similar to theirs (i.e., are within a confidence bound). In this paper, we formulate and analyze discrete-time BCMs with heterogeneous and adaptive confidence bounds. We introduce two new models: (1) a BCM with synchronous opinion updates that generalizes the Hegselmann--Krause (HK) model and (2) a BCM with asynchronous opinion updates that generalizes the Deffuant--Weisbuch (DW) model. We analytically and numerically explore our adaptive BCMs' limiting behaviors, including the confidence-bound dynamics, the formation of clusters of nodes with similar opinions, and the time evolution of an ``effective graph'', which is a time-dependent subgraph of a network with edges between nodes that are currently receptive to each other. For a variety of networks and a wide range of values of the parameters that control the increase and decrease of confidence bounds, we demonstrate numerically that our adaptive BCMs result in fewer major opinion clusters and longer convergence times than the baseline (i.e., nonadaptive) BCMs. We also show that our adaptive BCMs can have adjacent nodes that converge to the same opinion but are not 
receptive to each other. This qualitative behavior does not occur in the associated baseline BCMs.
\end{abstract}

% REQUIRED

% REQUIRED

%%%%%%

% -------------------------------------------
% -------------------------------------------
\section{Introduction}\label{sec:introduction}
Social interactions play an important role in shaping the opinions of individuals, communities of people, and society at large \cite{bak2021}. An individual's opinion on a topic is often influenced by the people with whom they interact \cite{jackson2008}, and researchers in many disciplines study such interactions and how they change opinions and actions~\cite{noorazar_recent_2020}. 
In an agent-based model of opinion dynamics, each agent represents an individual and a network encodes which agents are able to interact with each other. Each node (i.e., agent) of a network has an opinion in some opinion space. Studying opinion models allows researchers to examine the evolution of opinions on social networks with time, leading to insights into the spread of ideas \cite{friedkin_social_1990,jia_opinion_2015}, 
when communities of individuals reach consensus and when they do not \cite{vasca2021}, and the formation of ``opinion clusters'' (i.e., clusters of nodes with similar opinions) \cite{lorenz2008}.

Individuals are often influenced most by people and other sources whose opinions are similar to theirs \cite{selective_exposure_def}.
This phenomenon is encapsulated in a simple form in \emph{bounded-confidence models} (BCMs) \cite{noorazar-review, HK_model, deffuant} of opinion dynamics, in which the nodes of a network have continuous-valued opinions and interacting nodes influence each others' opinions if and only if their opinions are sufficiently similar. A key feature of BCMs is the presence of a ``confidence bound'', which is a parameter that determines which nodes can influence each other. A node can only influence and be influenced by its neighbors when the difference in their opinions is less than their confidence bound. 

The two most popular BCMs are the Hegselmann--Krause (HK) model \cite{krause2000, HK_model} and the Deffuant--Weisbuch (DW) model \cite{deffuant}. Both of these models use discrete time.
The HK model updates synchronously; at each time, every node updates its opinions based on the opinions of all of its neighbors. 
The DW model updates asynchronously; at each time, one selects a dyad (i.e., a pair of adjacent nodes and the edge between them), and the two nodes in the dyad interact and potentially influence each others' opinions. The DW model also has a compromise parameter, which controls how much nodes in a dyad influence each other when they compromise their opinions. In both the HK model and the DW model, the confidence bound is traditionally a single fixed scalar parameter that is the same for all dyads. 
In \cref{sec:related}, we discuss generalizations of these baseline models.

In the present paper, we formulate and study adaptive-confidence BCMs that generalize the HK and DW models by incorporating distinct, time-dependent confidence bounds for each dyad.
The choice of modeling interactions synchronously (as in the HK model) or asynchronously (as in the DW model) impacts a BCM's tendency towards consensus \cite{urbig2008}.
Because the synchronous updates of the HK model yield faster convergence times than asynchronous updates, they allow us to more feasibly study our adaptive-confidence HK model on larger networks than for our adaptive-confidence DW model. Therefore, we concentrate more on our adaptive-confidence HK model than on our adaptive-confidence DW model, although we do discuss some analytical and computational results for the latter.

The confidence bounds in our adaptive-confidence BCMs change after nodes interact with each other. These changes highlight the idea that the quality of an interaction between individuals can affect how much they trust each other~\cite{glanville_social_connections_create_trust,choi_market_2020,li_who_2017}.
For example, in online marketplaces, trust between users depends on their past experiences with each other and on the reported experiences of other users in reputation systems \cite{sherchan2013,ruohomaa2007,resnick2000}.
The word ``trust'' can have different meanings in different disciplines; one interpretation is that trust represents an expectation about future behavior~\cite{sherchan2013}. 
Rather than considering trust, our BCMs use a notion of ``receptiveness'', which encodes the willingness of an individual to consider the future opinions of another individual. 
When two nodes interact with each other, their mutual receptiveness changes. 
See~\cite{bagnoli2007,xiong2017,nugent2023} for other opinion models with interaction-influenced receptiveness. 

In our adaptive-confidence BCMs, when two nodes successfully compromise their opinions in an interaction (i.e., they have a ``positive interaction''), they become more receptive to each other. 
Likewise, when two nodes interact but do not change their opinions (i.e., they have a ``negative interaction''), they become less receptive to each other. 
When nodes $i$ and $j$ interact and influence each others' opinions (i.e., their current opinion difference is smaller than their current confidence bound), we increase their confidence bound $c_{ij}$. When nodes $i$ and $j$ interact and do not influence each others' opinions (i.e., their current opinion difference is {at least as large as}
their current confidence bound), we decrease their confidence bound $c_{ij}$. In our adaptive-confidence BCMs, {each dyad has a distinct confidence bound and} interactions are symmetric (i.e., either both nodes influence each other or neither node influences the other).
One can interpret the increase of a dyadic confidence bound in our BCMs as nodes becoming more receptive to nodes with whom they compromise, and one can interpret the decrease of a dyadic confidence bound as nodes becoming less receptive to nodes with whom they do not compromise. 
When nodes in our BCMs have a negative interaction, they adapt their dyadic confidence bounds, but their opinions stay the same. Other researchers have considered BCMs with ``repulsion'', in which the opinions of interacting nodes with sufficiently different opinions move farther apart from each other~\cite{alizadeh2015,huet_rejection_2008,kann2023}.

In the present paper, we study the time evolution and long-term behaviors of our adaptive-confidence HK and DW models. We examine the
formation of ``limit opinion clusters'' (i.e., sets of nodes that converge to the same opinion),
the dynamics of the confidence bounds, and the convergence rate of the opinions. 
We simulate our models on various networks (see \cref{sec:simulations_networks}) and study the time evolution of their associated ``effective graphs'', which are time-dependent subgraphs of a network with edges only between nodes {that are receptive to each other} (see \cref{sec:theory}). 
We show numerically that our adaptive-confidence BCMs tend to have less ``opinion fragmentation''\footnote{In our study, ``opinion fragmentation'' signifies the existence of at least two ``major'' opinion clusters, which include
more than 1\% of the nodes of a network. 
In \cref{sec:simulations_quantities}, we give more detail about how we define opinion fragmentation and major opinion clusters.
} 
than their associated baseline (i.e., nonadaptive) BCMs. 
In our numerical simulations, we study ``final'' opinion clusters (see \cref{sec:simulations_specs}) to approximate limit opinion clusters.
We demonstrate numerically that the connected components of the final effective graphs in our BCMs can have more complicated structures than those of the baseline BCMs.

%%%%

% -----------------------------------------
\subsection{Related work}\label{sec:related}

There has been much research on standard (i.e., nonadaptive) HK and DW models on networks through numerical simulations \cite{meng_opinion_2018,fortunato_universality_2004,lorenz_consensus_2006,HK_model} 
and both heuristic analytical arguments and mathematically rigorous proofs \cite{ben-naim_bifurcations_2003,lorenz_stabalization_2005,lorenz2008,HK_model}.
The DW model was studied initially on both a fully-mixed population (i.e., a complete network) and on a square-lattice network~\cite{deffuant}, and the HK model was studied initially only on a fully-mixed population~\cite{HK_model}. 
Subsequently, the DW and HK models have been studied on a variety of networks~\cite{fortunato2005,meng_opinion_2018,schawe2021_HK_networks}.
See \cite{noorazar-review, noorazar_recent_2020, bernardo2024} for reviews of research on the standard DW and HK models and their generalizations.

Many researchers have generalized the HK and DW models by incorporating heterogeneity into the confidence bounds.
Lorenz \cite{lorenz_heterogeneous_2009} extended these BCMs so that each node has its own confidence bound, which can result in asymmetric influence and opinion updates. Using numerical simulations, Lorenz demonstrated that these BCMs are more likely than the baseline BCMs to reach a consensus state when there are both open-minded and close-minded nodes
(which have large and small confidence bounds, respectively).
By analyzing the heterogeneous-confidence DW model of \cite{lorenz_heterogeneous_2009} on {complete graphs,} 
Chen et al.~\cite{chen_convergence_2020} 
proved almost-sure convergence of opinions for certain parameter values and derived sufficient conditions for the nodes of a network to eventually reach a consensus.
In a related work, Chen et al. \cite{chen_heterogeneous_2020} 
examined a heterogeneous HK model with ``environmental noise'' (e.g., from media sources) and showed that heterogeneous confidence bounds in this setting can yield larger differences in node opinions in the infinite-time limit.
Su et al.~\cite{su_partial_2017} examined the heterogeneous-confidence HK model of \cite{lorenz_heterogeneous_2009} and proved that at least some nodes of a network converge to a steady-state opinion in finite time. 

Researchers have also incorporated edge-based heterogeneities in the confidence bounds of BCMs.
Etesami~\cite{etesami2019} examined an HK model on networks with time-independent edge-heterogeneous confidence bounds {and proved that their model is Lyapunov stable.}
Shang~\cite{shang_agent_2014} studied a DW model in which 
each edge has a confidence bound that takes a value from
an independent and identically distributed Poisson process. 
They derived sufficient conditions for consensus to occur almost surely for a one-dimensional lattice graph.

Other generalizations of BCMs and related opinion models generalize the model parameters by making them time-dependent or adaptive.
Weisbuch et al.~\cite{weisbuch_meet_2002} studied a generalized DW model in which each
{node} 
has a heterogeneous, time-dependent confidence bound that is proportional to the standard deviation of the opinions that that {node} 
observed in all prior interactions. 
They also considered a variant of their model that places more weight on the observed opinions from recent interactions.
Deffuant et al.~\cite{deffuant_how_2002} examined a DW model with ``relative agreement''. In their model, each {node} has an uncertainty parameter that determines (1) whether it and the {node} with which it interacts influence each other and (2) the amount by which they influence each other.
A node changes both its opinion and its uncertainty when {another node influences it.}
Bagnoli et al.~\cite{bagnoli2007} considered a BCM on complete graphs in which each pair of adjacent nodes (i.e., each dyad)
has an associated time-dependent affinity value (which determines whether or not they can influence each other) that depends on the magnitude of their opinion difference.
Chacoma and Zanette~\cite{chacoma2015} examined opinion and confidence changes in a questionnaire-based experiment,
and they then proposed an agent-based opinion model based on the results of their experiment.
Their model is not a BCM, but it does incorporate a notion of time-dependent confidence between nodes.
Bernardo, Vasca, and Iervolino~\cite{vasca2021, bernardo2022} developed variants of the HK model in which nodes have individual, time-dependent confidence bounds\footnote{{The confidence bounds update with time in different ways in the models in \cite{vasca2021} and \cite{bernardo2022}.}} that depend on the opinions of neighboring nodes.
In their models, nodes adapt their confidence bounds through a heterophilic mechanism (i.e., they seek neighboring nodes whose opinions differ from theirs).
By contrast, in our models, nodes do not actively seek neighbors with different opinions. Instead, their mutual receptiveness increases when their opinions are sufficiently close to each other.

In our paper, we incorporate adaptivity into the confidence bounds of BCMs, but one can instead incorporate adaptivity in the network structures of BCMs \cite{kozma_consensus_2008,kozma_consensus_2008-1,del_vicario_modeling_2017,kan_adaptive_2022}.\footnote{See the reviews \cite{adaptivity2023,berner2023} for discussions of various notions of adaptivity in dynamical systems.}
Kozma and Barrat \cite{kozma_consensus_2008,kozma_consensus_2008-1} modified the DW model to allow rewiring of ``discordant'' edges, which occur between nodes whose opinions differ from each other by more than the confidence bound. In their model, rewired edges connect to new nodes uniformly at random. 
Recently, Kan et al.~\cite{kan_adaptive_2022} generalized this model by including both a confidence bound and an opinion-tolerance threshold, with discordant edges occurring between nodes whose opinions differ by more than that threshold. 
They incorporated opinion homophily into the rewiring probabilities, so nodes are more likely to rewire to nodes with more similar opinions.
They observed in numerical simulations that it is often harder to achieve consensus in their adaptive DW model than in an associated baseline DW model.

There has been much theoretical development of models of opinion dynamics, and it is 
important to empirically validate these models \cite{vazquez2022, galesic2021}.
Some researchers have used questionnaires \cite{chacoma2015,vandeKerckhove2016,takacs2016} or data from social-media platforms \cite{kozitsin2022_real_data, kozitsin2023_micro_level} to examine how opinions change in controlled experimental settings. 
Another approach is to develop models of opinion dynamics that infer model parameters \cite{chu2022, kozitsin_linking_2022} or opinion trajectories \cite{monti2020} from empirical data. 
There are many challenges to developing and validating models of opinion dynamics that represent real-world situations \cite{mas2019, bak2021}, but mechanistic modeling is valuable, as it (1) forces researchers to clearly
specify relationships and assumptions during model development and (2) provides a framework to explore complex social phenomena \cite{vazquez2022, holme2015}.

% -------------------------------------------
\subsection{Organization of our paper}\label{sec:organization}
Our paper proceeds as follows. In \cref{sec:model-introduce}, we introduce our adaptive-confidence BCMs and discuss the associated baseline BCMs. In \cref{sec:theory}, we give theoretical guarantees for our adaptive-confidence BCMs. We describe the specifications of our numerical simulations in \cref{sec:simulations} and the results of our numerical simulations in \cref{sec:results}. In \cref{sec:discussion}, we summarize our main results and discuss possible avenues of future work. 
In \cref{appendix:proofs} and \cref{appendix:proof_eff_baseline_DW}, 
we prove the results of \cref{sec:adaptive-confidenceDW_theorems}.
We present additional numerical results for our adaptive-confidence HK model in \cref{appendix:HK_simulations}, and we present additional numerical results for our adaptive-confidence DW model in \cref{appendix:DW_simulations}.
Our code and plots are available 
at \url{https://gitlab.com/graceli1/Adaptive-Confidence-BCM}.

% -------------------------------------------
% -------------------------------------------
\section{Baseline and adaptive BCMs}\label{sec:model-introduce}

We extend the HK and DW models by introducing adaptive confidence bounds. For both the HK and DW models, which we study on networks, we first present the baseline BCM and then introduce our adaptive-confidence generalization of it. The nodes in our BCMs 
represent agents that
have opinions that lie in the closed interval $[0,1]$. 
Let $G = (V,E)$, where $V$ is the set of nodes and $E$ is the set of edges, denote a time-independent, unweighted, and undirected graph without self-edges or multi-edges. The edges in the set $E$ specify which pairs of nodes can interact with each other at each discrete time $t$.
Let $N = |V|$ denote the number of nodes of the graph (i.e., network), $x_i(t)$ denote the opinion of node $i$ at time $t$, and $\vec{x}(t)$ denote the vector of the opinions of all nodes at time $t$ (i.e., the entry $[\vec{x}(t)]_i = x_i(t)$). We denote the edge that is attached to adjacent nodes $i$ and $j$ by $(i,j)$.

% -------------------------------------------
\subsection{The HK model}\label{sec:originalHK}
The baseline HK model \cite{krause2000, HK_model} is a discrete-time synchronous BCM on 
a time-independent, unweighted, and undirected graph $G = (V,E)$
with no self-edges or multi-edges.\footnote{The HK model was examined initially on a fully-mixed population \cite{HK_model}, but we use its extension to networks (see, e.g., \cite{fortunato2005, parasnis_hegselmann-krause_2018, schawe2021_HK_networks})
as our ``baseline HK model''.}
At each time $t$, we update the opinion of each node $i$ by calculating
\begin{equation}\label{eq:HK_update_rule}
	x_{i}(t + 1) = |I(i, x(t))|^{-1} \sum_{j \in I(i, x(t))} x_{j}(t)\,,
\end{equation}
where\footnote{In \cite{krause2000, HK_model}, 
$I(i, x(t)) = \{ i \} \cup \left\{j \ | \ |x_i(t) - x_j(t)| \leq c { \text{ and }(i,j)\in E} \right\}$. 
We use a strict inequality to be consistent with the strict inequality in the DW model.}
$I(i, x(t))= \{ i \} \cup \left\{j \ | \ |x_i(t) - x_j(t)| < c { \text{ and }(i,j)\in E} \right\} \subseteq \left\{1, 2, \ldots, N \right\}$. 
The \emph{confidence bound} {$c \in [0,1]$} 
controls the ``open-mindedness'' of nodes to different opinions.
We say that adjacent nodes $i$ and $j$ are 
\emph{receptive} to each other at time $t$ if their opinion difference is less than the confidence bound $c$ (i.e., $|x_i(t) - x_j(t)| < c$).
Accordingly, $I(i,x(t))$ is node $i$ itself along with all adjacent nodes 
to which $i$ is receptive.\footnote{An alternative interpretation is that each node has a self-edge. We do not use this interpretation.}
The confidence bound $c$ in the baseline HK model is homogeneous (i.e., the confidence bound is the same for all dyads) and time-independent.

% -------------------------------------------
\subsection{Our HK model with adaptive confidence bounds}\label{sec:DynamicHK}

Our HK model with adaptive confidence bounds is similar to the baseline HK model 
{with update rule} \cref{eq:HK_update_rule}, 
but now each edge $(i,j) \in E$ has a dyadic confidence bound $c_{ij}(t) \in [0,1]$ that is time-dependent and changes after each interaction between the nodes {in that dyad}. We refer to this model as our \emph{adaptive-confidence HK model}. 
Instead of a fixed confidence bound,
there is an initial confidence bound {$c_0 \in (0,1)$} and we initialize all of the confidence bounds\footnote{When $c_0 = 0$, nodes are never receptive to their neighbors (i.e., $c_{ij}(t) = 0$ for all adjacent nodes $i$ and $j$ at all times $t$).
When $c_0 = 1$, all nodes are always receptive to all of their neighbors (i.e., $c_{ij}(t) = 1$ for all adjacent nodes $i$ and $j$ at all times $t$).
We do not examine these values of $c_0$.}
to $c_{ij}(0) = c_0$ for each edge $(i,j) \in E$.
There is also a confidence-increase parameter $\gamma \in [0,1]$ and a confidence-decrease parameter $\delta \in [0,1]$, which control how much $c_{ij}(t)$ increases and decreases, respectively, after each interaction. 

At each time $t$, we update the opinion of each node $i$ by calculating
\begin{equation}\label{eq:HK_opinion}
	x_{i}(t+1) = |I(i, x(t))|^{-1} \sum_{j \in I(i, x(t))} x_{j}(t) \, ,
\end{equation}
where\footnote{{Although equations \cref{eq:HK_update_rule} and \cref{eq:HK_opinion} look the same, they use different definitions of the quantity $I(i,x(t))$. Equation \cref{eq:HK_update_rule} has a homogeneous and 
time-independent confidence bound, whereas equation \cref{eq:HK_opinion} has heterogeneous and adaptive confidence bounds.}} $I(i, x(t)) =
\{ i \} \cup \left\{j \ | \ |x_i(t) - x_j(t)| < c_{ij}(t){ \text{ and }(i,j) \in E}\right\} \subseteq \left\{1, 2, \ldots, N \right\}$.
Adjacent nodes $i$ and $j$ are 
receptive to each other at time $t$ if their opinion difference is less than their dyadic confidence bound $c_{ij}$ (i.e., $|x_i(t) - x_j(t)| < c_{ij}(t)$).
At each time, we also update each confidence bound $c_{ij}$ by calculating
\begin{equation}\label{eq:HK_confidence}
	c_{ij}(t+1)
    = \begin{cases}
        c_{ij}(t) + \gamma(1-c_{ij}(t)) \,, 
        & \text{ if } |x_i(t) - x_j(t)| < c_{ij}(t) \\ 
        \delta c_{ij}(t) \,,  
        & \text{ if } |x_i(t) - x_j(t)| \geq c_{ij}(t) \,.
	\end{cases}
\end{equation}
That is, if the opinion difference between {adjacent} nodes $i$ and $j$ is smaller than their confidence bound at time $t$, their associated dyadic confidence bound $c_{ij}$ increases; otherwise, their dyadic confidence bound decreases. 
Larger values of $\gamma$ correspond to sharper increases in the receptiveness between nodes when nodes compromise their opinions. Smaller values of $\delta$ correspond to sharper drops in the receptiveness between nodes
when nodes interact but do not compromise. 

Because {$c_0 \in (0,1)$ and $\gamma, \delta \in [0,1]$}, 
the update rule \cref{eq:HK_confidence} preserves {$c_{ij}(t) \in (0,1)$} for each edge $(i,j)$ and all times $t$.
If $(\gamma, \delta) = (0,1)$,
then $c_{ij}(t) = c_0$ for all $t$ and all edges $(i,j)\in E$.
That is, the confidence bounds are homogeneous and time-independent, so our adaptive-confidence HK model reduces to the baseline HK model.

% -------------------------------------------
\subsection{The DW model}\label{sec:originalDW}
The baseline DW model \cite{deffuant} is a discrete-time asynchronous BCM on a time-independent, 
unweighted, and undirected graph $G = (V,E)$ with no self-edges or multi-edges. At each time $t$, we choose an edge $(i,j) \in E$ uniformly at random. 
If nodes $i$ and $j$ are receptive to each other (i.e., if 
the opinion difference $|x_i(t) - x_j(t)|$ is less than the confidence bound $c$), 
we update the opinions of these nodes by calculating
\begin{equation}\label{eq:DW_update_rule}
\begin{split}
	x_i(t+1) &= x_i(t) + \mu(x_j(t) - x_i(t)) \,,\\
	x_j(t+1) &= x_j(t) + \mu(x_i(t) - x_j(t)) \,, 
\end{split}
\end{equation}
where $\mu \in (0, 0.5]$ is the compromise parameter\footnote{Alternatively, one can consider $\mu \in (0,1)$ as in Meng et al. \cite{meng_opinion_2018}, although this is an uncommon choice. When $\mu > 0.5$, nodes ``overcompromise'' when they change their opinions; they overshoot the mean opinion and change which side of the mean opinion they are on.
}.
Otherwise, the opinions $x_i$ and $x_j$ remain the same. At a given time $t$, we do not update the opinions of any nodes other than $i$ and $j$.
The confidence bound {$c \in [0,1]$ in the baseline DW model}
is homogeneous (i.e., the confidence bound is the same for all dyads) and time-independent.
As in the HK model, the confidence bound $c$ controls the open-mindedness of nodes to different opinions.
The compromise parameter $\mu$ indicates how much nodes adjust their opinions when they interact with a node to whom they are receptive.
When $\mu = 0.5$, two interacting nodes that are receptive to each other precisely average their opinions; when $\mu \in (0,0.5)$, interacting nodes that are receptive to each other move towards each others' opinions, but they do not adopt the mean opinion. Unlike in the HK model, the asynchronous update rule \eqref{eq:DW_update_rule} of the DW model incorporates only pairwise opinion updates.

% -------------------------------------------
\subsection{Our DW model with adaptive confidence bounds}\label{sec:dynamicDW}
We refer to our DW model with adaptive confidence bounds as our \emph{adaptive-confidence DW model}. As in the baseline DW model, there is a compromise parameter $\mu \in (0,0.5]$. 
As in our adaptive-confidence HK model, we initialize the confidence bounds in our adaptive-confidence DW model to be $c_{ij}(0) = c_0$, where $c_0 \in (0,1)$ is the initial confidence bound.\footnote{As in our adaptive-confidence HK model, when $c_0 = 0$, nodes are never receptive to any of their neighbors (i.e., $c_{ij}(t) = 0$ for all adjacent nodes $i$ and $j$ at all times $t$).
When $c_0 = 1$, nodes are always receptive to all of their neighbors (i.e., $c_{ij}(t) = 1$ for all adjacent nodes $i$ and $j$ at all times $t$).
We do not examine these values of $c_0$.}
There again is a confidence-increase parameter $\gamma \in [0,1]$ and a confidence-decrease parameter $\delta \in [0,1]$, 
which control how much $c_{ij}(t)$ increases and decreases, respectively, after each interaction.

At each time $t$, we select an edge $(i,j) \in E$ uniformly at random. 
If nodes $i$ and $j$ are receptive to each other (i.e., if $|x_i(t) - x_j(t)| < c_{ij}(t)$), we update the opinions of nodes $i$ and $j$ using the DW update rule \cref{eq:DW_update_rule}.
Otherwise, the opinions $x_i$ and $x_j$ remain the same. 
We also update the dyadic confidence bound $c_{ij}$ using equation \cref{eq:HK_confidence}.
That is, if the opinions of nodes $i$ and $j$ differ by less than their current dyadic confidence bound at time $t$, the confidence bound increases; otherwise, it decreases. 
The update rules preserves $c_{ij}(t) \in (0,1)$ for each edge $(i,j)$ and all times $t$.
All other opinions and confidence bounds remain the same.
Our adaptive-confidence DW model reduces to the baseline DW model when $(\gamma, \delta) = (0, 1)$.

% -------------------------------------------
% -------------------------------------------
\section{Theoretical results}\label{sec:theory}
We now discuss some theoretical guarantees of our BCMs. 
We study them on time-independent, unweighted, and undirected graphs without self-edges or multi-edges.

As a consequence of \cref{thm:lorenzthm} (which we state shortly), the opinion of each node in our BCMs converges to some limit value.
We define the \emph{limit opinion} $x^i$ of node $i$ as $\lim\limits_{t\to\infty}x_i(t)$.
We say that nodes $i$ and $j$ are in the same \emph{limit opinion cluster} if 
\begin{equation}\label{eq:opinion_cluster}
	\lim_{t\to\infty} x_i(t) = \lim_{t\to\infty} x_j(t)\,.
\end{equation}
Therefore, equation \cref{eq:opinion_cluster} gives an equivalence relation on the set of nodes; 
the limit opinion clusters are the equivalence classes.

A graph $G = (V,E)$ in a BCM has an associated time-dependent \emph{effective graph} $G_\mathrm{eff}(t)$, which is a subgraph of $G$ with edges only between nodes that 
are receptive to each other at time $t$.\footnote{Other researchers have referred to the effective graph as a ``confidence graph'' \cite{bernardo2024}, a ``communication graph'' \cite{bhattacharyya2013}, and a ``corresponding graph'' \cite{yang2014}.} 
That is,
\begin{align}
	G_\mathrm{eff}(t) &= (V,E_\mathrm{eff}(t))\,, \label{eq:eff_graph} \\
	E_\mathrm{eff}(t) &= \{(i,j)\in E \text{ such that } |x_i(t) - x_j(t)| < c_{ij}(t)\} \nonumber \, .
\end{align}

Consider the following theorem, which was stated and proved by Lorenz \cite{lorenz_stabalization_2005}.

\begin{thm}\label{thm:lorenzthm}
Let $\{ A(t) \}_{t=0}^\infty \in \mathbb{R}_{\geq 0}^{N \times N}$ be a sequence of row-stochastic matrices. Suppose that each matrix satisfies the following properties:

\begin{enumerate}
\item[(1)] The diagonal entries of $A(t)$ are positive.
\item[(2)] For each $i, j \in \{1,\ldots,N\}$, we have
 $[A(t)]_{ij} > 0$ if and only if $[A(t)]_{ji} > 0$.
\item[(3)] 
There is a constant $\alpha > 0$ such that the smallest positive entry of $A(t)$ for each $t$ is larger than $\alpha$.
\end{enumerate}

Given times $t_0$ and $t_1$ with $t_0 < t_1$, let 
\begin{equation}
	A(t_0, t_1) = A(t_1-1)\times A(t_1-2)\times\cdots\times A(t_0)\,.
\end{equation}
If conditions (1)--(3) are satisfied, then there exists a time $t'$ and pairwise-disjoint classes $\mathcal{I}_{1} \cup \cdots \cup \mathcal{I}_{p} = \{1,\ldots,N\}$ such that if we reindex the rows and columns of the matrices in the order $\mathcal{I}_{1}, \ldots, \mathcal{I}_{p}$, then
\begin{equation}
	\lim _{t \rightarrow \infty} A(0, t) = \left[\begin{array}{ccc}
		K_{1} & & 0 \\
		& \ddots & \\
		0 & & K_{p}
	\end{array}\right] A\left(0, t'\right)\,,
\end{equation}
where each $K_q$, with $q \in \{1, 2, \dots , p\}$, is a row-stochastic matrix of size $|\mathcal{I}_{q}|\times |\mathcal{I}_{q}|$ whose rows are all the same.
\end{thm}

As stated in \cite{lorenz_stabalization_2005},
\cref{thm:lorenzthm} guarantees that the opinion of each node convergences to a limit opinion in the baseline HK and DW models. Because the node opinions in our adaptive-confidence HK and DW models update in the same way as in the corresponding baseline BCMs (see \cref{eq:HK_update_rule} and \cref{eq:DW_update_rule}, respectively), it follows that the node opinions in our models also converge to a limit opinion.

% -------------------------------------------
\subsection{Adaptive-confidence HK model}\label{sec:HK_theory}
\subsubsection{Confidence-bound analysis}\label{sec:conf_bound_analysis_HK}
In \cref{thm:convergence_HK}, we give our main result about the behavior of the confidence bounds (which update according to \cref{eq:HK_confidence}) in our adaptive-confidence HK model.

\begin{thm}\label{thm:convergence_HK}
In our adaptive-confidence HK model (with update rules \cref{eq:HK_opinion} and \cref{eq:HK_confidence}) with parameters $\gamma \in (0,1]$ and $\delta \in [0,1)$, the dyadic confidence bound $c_{ij}(t)$ of each pair of adjacent nodes, $i$ and $j$, converges either to $0$ or to $1$. Furthermore, if $i$ and $j$ are in different limit opinion clusters, then $c_{ij}(t)$ converges to $0$.
\end{thm}

We prove \cref{thm:convergence_HK} by proving \cref{lemma:no_influence_HK},
\cref{lemma:monotone_HK} and \cref{lemma:convergence_HK}, which we state shortly. Because $c_{ij}(t) \in [0,1]$, \cref{lemma:monotone_HK} gives convergence (because an eventually monotone\footnote{We say that a discrete time series $a(t)$ is ``eventually monotone increasing'' (respectively, ``eventually monotone decreasing'') if there exists a time $T$ such that $a(t + 1) \geq a(t)$ (respectively, $a(t + 1) \leq a(t)$) for all $t \geq T$. 
Additionally, we say that $a(t)$ is ``eventually strictly increasing'' (respectively, ``eventually strictly decreasing'') if there exists a time $T$ such that $a(t + 1) > a(t)$ (respectively, $a(t + 1) < a(t)$) for all times $t \geq T$.
\label{footnote:monotone_strict}
}
sequence in $[0,1]$ must converge). By \cref{lemma:convergence_HK}, we then have convergence either to $0$ or to $1$. Furthermore, by \cref{lemma:monotone_HK}, if nodes $i$ and $j$ are in different limit opinion clusters, then $c_{ij}(t)$ must eventually be strictly decreasing and hence must converge to $0$.
However, if adjacent nodes $i$ and $j$ are in the same limit opinion cluster, then $c_{ij}$ does not necessarily converge to $1$. In fact, as we discuss in \cref{sec:results}, our numerical simulations suggest that it is possible for the confidence bound of adjacent nodes in the same limit opinion cluster to instead converge to $0$.

In \cref{thm:convergence_HK}, \cref{lemma:monotone_HK}, and \cref{lemma:convergence_HK},
we consider our adaptive-confidence HK model with parameters $\gamma \in (0,1]$ and $\delta \in [0,1)$. These parameter restrictions preclude the baseline HK model (which is equivalent to our adaptive-confidence HK model with $(\gamma, \delta) = (0,1)$).
However, in \cref{lemma:no_influence_HK}, we consider
 $\gamma \in [0,1]$ and $\delta \in [0,1]$. Therefore, \cref{lemma:no_influence_HK} also applies to the baseline HK model, 
so we use it in our proof of \cref{thm:effgraph_baseline_HK} for the baseline HK model.

\begin{lemma}\label{lemma:no_influence_HK}
Consider our adaptive-confidence HK model (with update rules \cref{eq:HK_opinion} and \cref{eq:HK_confidence}) with {parameters} $\gamma \in [0,1]$ and $\delta \in [0,1]$.
There is a time $T$ such that no adjacent nodes $i$ and $j$ in different limit opinion clusters (i.e., $x^i \neq x^j$) are receptive to each other at any time $t \geq T$.
That is, for all pairs of adjacent nodes $i$ and $j$ with $x^i \neq x^j$, we have
$|x_i(t) - x_j(t)| \geq c_{ij}(t)$ for all times $t \geq T$.
\end{lemma}

\begin{proof}

Consider a pair of adjacent nodes, $i$ and $j$, that are in different limit opinion clusters (i.e., $x^i \neq x^j$). 
Let $d$ be $1$ more than the largest degree of a node of the graph $G$; that is, $d = 1 + \displaystyle\max_{i\in V} \deg(i)$. 
Choose $T$ such that the inequalities 
\begin{align}
	|x_k(t) - x^k| &< \frac{1}{4d}\min_{x^m\neq x^{n}}|x^m - x^{n}| \,, \label{eq:HK_monotone_condition1} \\
	|x_k(t) - x_k(t')| &< \frac{1}{4d} \min_{x^m\neq x^{n}}|x^m - x^{n}| 
\label{eq:HK_monotone_condition2}
\end{align}
hold for each node $k$ and for all $t' > t \geq T$.

We claim that nodes $i$ and $j$ are not receptive to each other
at any time $t \geq T$. 
Suppose the contrary. There then must exist some time $t \geq T$ and adjacent nodes $i$ and $j$ with $x^i \neq x^j$ and $|x_i(t) - x_j(t)| < c_{ij}(t)$.
Fix such a {time} $t$ and choose a node $i$ that gives the smallest limit opinion value $x^i$ such that there is a neighboring node $j$ with
$x^j \neq x^i$ and $|x_i(t) - x_j(t)| < c_{ij}(t)$.

For this node $i$, let $q = |I(i,x(t))| \leq d$.
Because of our choice of $x^i$, we have 
\begin{equation}\label{eq:HK_proof_limit_opinion_difference_bound}
	\frac{1}{q}\left|\sum_{\substack{j\in I(i,x(t))\\ j\neq i}} (x^i - x^j) \right|
	=\frac{1}{q}\sum_{\substack{j\in I(i,x(t))\\ j\neq i}} (x^j - x^i)
\geq\frac{1}{d}\min_{x^m\neq x^{n}}|x^m - x^{n}| \,.
\end{equation}
Using \cref{eq:HK_monotone_condition1}, we obtain
\small
\begin{align}
	\frac{1}{q}\left| \sum_{\substack{j\in I(i,x(t))\\ j\neq i}} (x^i - x^j) \right|
&\leq \frac{1}{q}\sum_{\substack{j\in I(i,x(t))\\ j\neq i}}\left|x_i(t) - x^i\right|
+ \frac{1}{q}\left|\sum_{\substack{j\in I(i,x(t))\\ j\neq i}} (x_i(t) - x_j(t)) \right| \nonumber\\
	&\qquad\qquad\qquad + \frac{1}{q}\sum_{\substack{j\in I(i,x(t))\\ j\neq i}}\left|x_j(t) - x^j\right| \nonumber \\
& < 2 \left(\frac{q-1}{q}\right) \left(\frac{1}{4d}\right) 
\min_{x^m\neq x^{n}}|x^m - x^{n}|
+\frac{1}{q}\left|\sum_{\substack{j\in I(i,x(t))\\ j\neq i}} (x_i(t) - x_j(t))\right| \nonumber \\
& < \frac{1}{2d}\min_{x^m\neq x^{n}}|x^m - x^{n}|
+\frac{1}{q}\left|\sum_{\substack{j\in I(i,x(t))\\ j\neq i}} (x_i(t) - x_j(t)) \right| \,. 
	\label{eq:HK_proof_limit_opinion_difference_other_one}
\end{align} 
\normalsize
Combining \cref{eq:HK_proof_limit_opinion_difference_bound} and \cref{eq:HK_proof_limit_opinion_difference_other_one} yields 
\begin{equation}\label{eq:HK_contradiction1}
	\frac{1}{q} \left|\sum_{\substack{j\in I(i,x(t))\\ j\neq i}} (x_i(t) - x_j(t)) \right| > \frac{1}{2d}\min_{x^m \neq x^{n}}|x^m - x^{n}| \,.
\end{equation}
Using the HK opinion-update rule \cref{eq:HK_opinion} and the inequality \cref{eq:HK_monotone_condition2}, we also have 
\begin{equation}\label{eq:HK_contradiction2}
	\frac{1}{q}\left|\sum_{\substack{j\in I(i,x(t))\\ j\neq i}} (x_i(t) - x_j(t)) \right|
		= |x_i(t+1)-x_i(t)|
		< \frac{1}{4d} \min_{x^m\neq x^{n}}|x^m - x^{n}| \,.
\end{equation}
The relations \cref{eq:HK_contradiction1} and \cref{eq:HK_contradiction2} cannot hold simultaneously, so nodes $i$ and $j$ 
are not receptive to each other at any time $t \geq T$.
\end{proof}

\begin{lemma}\label{lemma:monotone_HK}
In our adaptive-confidence HK model (with update rules \cref{eq:HK_opinion} and \cref{eq:HK_confidence}) with parameters $\gamma \in (0,1]$ and $\delta \in [0,1)$, the dyadic confidence bound $c_{ij}(t)$ of each pair of adjacent nodes, $i$ and $j$, is eventually either strictly increasing or strictly decreasing. 
That is, there is a time $T$ such that exactly one of the inequalities $c_{ij}(t_1) < c_{ij}(t_2)$ and $c_{ij}(t_1) > c_{ij}(t_2)$ holds for all times $t_2 > t_1 \geq T$.
Furthermore, if nodes $i$ and $j$ are in different limit opinion clusters {(i.e., if $x^i \neq x^j$)}, 
then $c_{ij}(t)$ is eventually strictly decreasing.
\end{lemma}

\begin{proof}

We first consider $c_{ij}$ for adjacent nodes $i$ and $j$ that are in different limit opinion clusters (i.e., $x^i \neq x^j$). By \cref{lemma:no_influence_HK}, there is a time $T$ such that 
nodes $i$ and $j$ are mutually unreceptive at all times $t \geq T$. 
Consequently, $c_{ij}(t)$ cannot increase at any time $t \geq T$ and must be monotone decreasing.
 
Because the adaptive-confidence HK model updates synchronously and the initial confidence bound $c_0 \in (0,1)$,
each confidence bound $c_{ij}$ must change at each time $t$. That is, for all pairs of adjacent nodes $i$ and $j$ and all times $t$, we have $c_{ij}(t + 1) \neq c_{ij}(t)$.
Consequently, for all adjacent nodes $i$ and $j$ in distinct limit opinion clusters and for all times $t \geq T$, we have that $c_{ij}$ is strictly decreasing (i.e., $c_{ij}(t + 1) < c_{ij}(t)$).

Now consider adjacent nodes $i$ and $j$ that are in the same limit opinion cluster (i.e., $x^i = x^j$). 
Choose $T > 0$ such that 
\begin{equation}\label{eq:HK_monotone_condition3}
    |x_k(t) - x^k| < \frac{\gamma}{2}
\end{equation}
for each node $k$ and all $t \geq T$.
We claim that there exists some $T_{ij} \geq T$ such that the dyadic confidence bound $c_{ij}$ is either strictly decreasing (i.e., $c_{ij}(t+1) < c_{ij}(t)$) or strictly increasing (i.e., $c_{ij}(t+1) > c_{ij}(t)$) for all times $t \geq T_{ij}$.

If $c_{ij}$ is strictly decreasing for all $t \geq T$, we choose $T_{ij} = T$. 
If $c_{ij}$ is not strictly decreasing for all $t \geq T$, there must exist some time $T_{ij} \geq T$ at which $|x_i(T_{ij}) - x_j(T_{ij})| < c_{ij}(T_{ij})$; therefore,
$c_{ij}(T_{ij} + 1) > c_{ij} (T_{ij})$. Without loss of generality, let $T_{ij}$ be the earliest such time. We will show by induction that $c_{ij}(t + 1) > c_{ij}(t)$ for all $t \geq T_{ij}$.
By assumption, this inequality holds for the base case $t = T_{ij}$.

Suppose that $c_{ij}(t + 1) > c_{ij}(t)$ for some value of $t \geq T_{ij}$.
We must then also have $|x_i(t) - x_j(t)| < c_{ij}(t)$ and 
\begin{equation}
	c_{ij}(t + 1) = c_{ij}(t) + \gamma(1 - c_{ij}(t)) \geq \gamma\,.
\end{equation}
By the inequality \cref{eq:HK_monotone_condition3}, we have
\begin{equation}
	|x_k(t') - x_{k'}(t')| \leq |x_k(t') - x^k| + |x^k - x^{k'}| + |x^{k'} - x_{k'}(t')| < \gamma
\end{equation}
for each node pair $k$ and $k'$ with $x^k = x^{k'}$ and all times $t' \geq T$.
Because $t+1 > T_{ij} \geq T$, it follows that
\begin{equation}
	|x_i(t+1) - x_j(t + 1)| < \gamma < c_{ij}(t + 1)\,,
\end{equation}
so $c_{ij}(t + 2) > c_{ij} (t + 1)$.
Consequently, by induction, if $c_{ij}$ increases at $t = T_{ij}$, then  $c_{ij}$ is strictly increasing (i.e., $c_{ij}(t+1) > c_{ij}(t)$) for all $t \geq T_{ij}$.
Therefore, there exists some time $T_{ij}$ such that $c_{ij}$ is either strictly decreasing or strictly increasing for all $t \geq T_{ij}$.

In summary, we have shown that $c_{ij}$ is eventually strictly decreasing for all adjacent nodes $i$ and $j$ in different limit opinion clusters.
Additionally, for all adjacent nodes $i$ and $j$ in the same limit opinion cluster, we have shown that $c_{ij}$ is eventually either strictly decreasing or strictly increasing.
\vspace{-18pt}
\[\] %The proof qed box is showing up on the last equation environment before the proof ends. We include a blank equation here to get the symbol to show up at the end.

\end{proof}

\begin{lemma}\label{lemma:convergence_HK}
Consider our adaptive-confidence HK model (with update rules \cref{eq:HK_opinion} and \cref{eq:HK_confidence}) with parameters $\gamma \in (0,1]$ and $\delta \in [0,1)$.
Suppose that $c^{ij} = \lim\limits_{t\to\infty} c_{ij}(t)$ exists.
It then follows that either $c^{ij} = 0$ or $c^{ij} = 1$. 
\end{lemma}

\begin{proof}
Given $\epsilon > 0$, choose a time $T$ so that the inequalities
\begin{align}
	|c_{ij}(t) - c^{ij}| &<  \epsilon/2\,, \\
	|c_{ij}(t_1) - c_{ij}(t_2)| &< \frac{1}{2}
    \left(\min\{1 - \delta, \gamma\}\right) \epsilon
\end{align}
hold for all times $t, t_1, t_2 \geq T$.
Fix some time $t \geq T$. It must be the case that either
\begin{equation}
	c_{ij}(t+1) = \delta  c_{ij}(t)
\end{equation}
or 
\begin{equation}
	c_{ij}(t+1) = c_{ij}(t) + \gamma  (1 - c_{ij}(t)) \,.
\end{equation}

Suppose first that $c_{ij}(t+1) = \delta c_{ij}(t)$. In this case, we claim that $c^{ij} = 0$. To verify this claim, first note that $c_{ij}(t) - c_{ij}(t + 1) = (1 - \delta)c_{ij}(t)$.
Because 
$c_{ij}(t) - c_{ij}(t + 1) < \frac{1}{2}(1 - \delta)\epsilon$, we see that $c_{ij}(t) < \epsilon/2$. Therefore,
\begin{align*}
	0 \leq c^{ij} &\leq |c^{ij} - c_{ij}(t)| + |c_{ij}(t)| \\
		&< \epsilon/2 + \epsilon/2\\
		&= \epsilon\,,
\end{align*}
which implies that $c^{ij} = 0$.

Now suppose that $c_{ij}(t+1) = c_{ij}(t) + \gamma(1 - c_{ij}(t))$. Note that $c_{ij}(t+1) - c_{ij}(t) = \gamma(1 - c_{ij}(t)) < \frac{1}{2}\gamma\epsilon$, which implies that $1 - c_{ij}(t) < \epsilon/2$. Additionally,
\begin{align*}
	0 \leq 1 - c^{ij} &\leq |1 - c_{ij}(t)| + |c_{ij}(t) - c^{ij}| \\ 
		&< \epsilon/2 + \epsilon/2 \\
		&= \epsilon\,,
\end{align*}
which implies that $c^{ij} = 1$. 

Therefore, it follows that either $c^{ij} = 0$ or $c^{ij} = 1$.
\end{proof}

% -------------------------------------------
\subsubsection{Effective-graph analysis}
In this section, we discuss the convergence of effective graphs in our adaptive-confidence HK model (see \cref{thm:effgraph_HK}) and the baseline HK model (see \cref{thm:effgraph_baseline_HK}). Our proofs of convergence employ some 
results from \cref{sec:conf_bound_analysis_HK}.

\begin{thm}\label{thm:effgraph_HK}
In our adaptive-confidence HK model with parameters $\gamma \in (0,1]$ and $\delta \in [0,1)$, the effective graph $G_\mathrm{eff}(t)$ is eventually constant with respect to time. That is, there is some time $T$ such that $G_\mathrm{eff}(t) = G_\mathrm{eff}(T)$ for all times $t \geq T$. Moreover, all of the edges of the \emph{limit effective graph} $\lim\limits_{t\to\infty} G_\mathrm{eff}(t)$ are between nodes in the same limit opinion cluster.
\end{thm}

\begin{proof}

By \cref{lemma:monotone_HK}, we can choose a time $T$ such that each dyadic confidence bound $c_{ij}$
is either strictly increasing or strictly decreasing for all times $t \geq T$.

For $t \geq T$, if $c_{ij}(t)$ is strictly decreasing, then we necessarily have that $|x_i(t) - x_j(t)| \geq c_{ij}(t)$ for all $t \geq T$, so $(i,j) \notin E_\mathrm{eff}(t)$ for all $t \geq T$. If $c_{ij}(t)$ is strictly increasing, then $|x_i(t) - x_j(t)| < c_{ij}(t)$ for all $t \geq T$, so $(i,j)\in E_\mathrm{eff}(t)$ for all $t \geq T$. Therefore, the set $E_\mathrm{eff}(t)$ of edges of the effective graph is constant for all $t \geq T$, so the effective graph is constant for $t \geq T$. 

For nodes $i$ and $j$ in different limit opinion clusters (i.e., $x^i\neq x^j$),
\cref{lemma:monotone_HK} guarantees that the confidence bound $c_{ij}$ is strictly decreasing for all times $t \geq T$.
Therefore, $|x_i(t) - x_j(t)| \geq c_{ij}(t)$ for all $t\geq T$, so $(i,j) \notin E_\mathrm{eff}(t)$ for all $t \geq T$.
\vspace{-22pt}
\[\] %The proof qed box is showing up on the last equation environment before the proof ends. We include a blank equation here to get the symbol to show up at the end.

\end{proof}

\cref{thm:effgraph_HK} states that all edges of a limit effective graph are between nodes in the same limit opinion cluster. 
However, the edges between nodes in the same limit opinion cluster do not have to exist in the limit effective graph. As we will discuss in \cref{sec:simulations_quantities} and \cref{sec:results}, our numerical simulations suggest that our adaptive-confidence BCMs can have adjacent nodes in the same limit opinion cluster whose associated dyadic confidence bound converges to $0$. 
The associated edge is thus not in the limit effective graph.

\cref{thm:effgraph_baseline_HK} guarantees that the effective graphs in the baseline HK model converge in the limit $t \rightarrow \infty$.
Unlike in our adaptive-confidence HK model, all edges between nodes in the same limit opinion cluster in the baseline HK model must exist in the limit effective graph.
Therefore, the limit opinion values in the baseline HK model fully determine the structure of the limit effective graph.

\begin{thm}\label{thm:effgraph_baseline_HK}
    In the baseline HK model (with update rule \cref{eq:HK_update_rule}), the effective graph $G_\mathrm{eff}(t)$ is eventually constant with respect to time. 
    {That is, there is some time $T$ such that $G_\mathrm{eff}(t) = G_\mathrm{eff}(T)$ for all times $t \geq T$.}
    Moreover, the edge $(i,j)\in E$ exists in the limit effective graph if and only if this edge is between two nodes in the same limit opinion cluster (i.e., $x^i = x^j$).
\end{thm}

\begin{proof}
We first consider adjacent nodes, $i$ and $j$, that are in different limit opinion clusters (i.e., $x^i \neq x^j$).
By \cref{lemma:no_influence_HK}, because our adaptive-confidence HK model with $\gamma = 0$ and $\delta = 1$ reduces to the baseline HK model, there exists a time $T_1$ such that nodes $i$ and $j$ are not receptive to
each other (i.e., $|x_i(t) - x_j(t)| \geq c$) 
{at} any time $t \geq T_1$. 
Therefore, the edge $(i,j) \notin E_\mathrm{eff}(t)$ 
{at any time} $t \geq T_1$.

Now consider adjacent nodes, $i$ and $j$, that are in the same limit opinion cluster (i.e., $x^i = x^j$).
Choose a time $T_2$ such that $|x_k(t) - x^k| < c/2$ for each node $k$ and all times $t \geq T_2$.
For all $t \geq T_2$, we then have
\begin{align*}
	|x_i(t) - x_j(t)| \leq |x_i(t) - x^i| + |x^i - x^j| + |x^j - x_j(t)| < c/2 + 0 + c/2 = c \,.
\end{align*}
Therefore, at any time $t \geq T_2$, nodes $i$ and $j$ 
are receptive to each other and the edge $(i,j) \in E_\mathrm{eff}(t)$.
By taking $T = \max \{T_1, T_2\}$, for any time $t \geq T$, we have that $(i,j) \notin E_\mathrm{eff}(t)$ for all edges $(i,j)$ with $x^i \neq x^j$ and that $(i,j) \in E_\mathrm{eff}(t)$ for all edges $(i,j)$ with $x^i = x^j$.

\vspace{-23pt}
\[\] %The proof qed box is showing up on the last equation environment before the proof ends. We include a blank equation here to get the symbol to show up at the end.

\end{proof} 

% -------------------------------------------
\subsection{Adaptive-confidence DW model}\label{sec:adaptive-confidenceDW_theorems}

In this section, we discuss our theoretical results for the confidence bounds and effective graphs in our adaptive-confidence DW model.
Both the baseline DW model and our adaptive-confidence DW model are asynchronous and stochastic. At each discrete time, we uniformly 
{randomly} select one pair of adjacent nodes to interact.
Because of the stochasticity in the baseline and adaptive-confidence DW models, our theoretical results for them are in an ``almost sure'' sense. 
By contrast, our theoretical results (see \cref{sec:HK_theory}) are deterministic for the baseline and adaptive HK models.

%%%

\subsubsection{Confidence-bound analysis}\label{sec:confidence_bound_analysis_DW}
We now give our main result about the behavior of the confidence bounds in our adaptive-confidence DW model (which has the update rules \cref{eq:DW_update_rule} and \cref{eq:HK_confidence}). This result mirrors the main result for our adaptive-confidence HK model in \cref{sec:conf_bound_analysis_HK}.

\begin{thm}\label{thm:convergence_DW}
In our adaptive-confidence DW model (with update rules \cref{eq:DW_update_rule} and \cref{eq:HK_confidence}) with parameters $\gamma \in (0,1]$ and $\delta \in [0,1)$, 
the dyadic confidence bound $c_{ij}(t)$ {of each pair of adjacent nodes, $i$ and $j$,}
converges either to $0$ or to $1$ almost surely. Moreover, if nodes $i$ and $j$ are in different limit opinion clusters (i.e., $x^i \neq x^j$), then $c_{ij}(t)$ converges to $0$ almost surely.
\end{thm}

We prove \cref{thm:convergence_DW} by proving \cref{lemma:monotone_DW} and \cref{lemma:convergence_DW}, which we state shortly and prove in \cref{appendix:proof_cij_DW}. 
Because our adaptive-confidence DW model updates asynchronously, \cref{lemma:monotone_DW} guarantees eventual monotone increase or decrease of the confidence bounds. This result differs from the eventual strict increase or decrease in the confidence bounds in our adaptive-confidence HK model (see \cref{lemma:monotone_HK}).
(See \cref{footnote:monotone_strict} for our usage of the terms ``eventual monotone increase'' (and decrease) and ``eventual strict increase'' (and decrease).)
Because $c_{ij}(t) \in [0,1]$, \cref{lemma:monotone_DW} implies convergence (because an eventually monotone sequence in $[0,1]$ must converge). By \cref{lemma:convergence_DW}, we have almost sure convergence to $0$ or to $1$. Moreover, if nodes $i$ and $j$ are in different limit opinion clusters, then \cref{lemma:monotone_DW} guarantees that $c_{ij}(t)$ is monotone decreasing.
By \cref{lemma:convergence_DW}, $c_{ij}(t)$ thus almost surely converges to $0$.

\begin{lemma}\label{lemma:monotone_DW}
In our adaptive-confidence DW model (with update rules \cref{eq:DW_update_rule} and \cref{eq:HK_confidence})
with parameters $\gamma \in (0,1]$ and $\delta \in [0,1)$, the dyadic confidence bound $c_{ij}(t)$ of each pair of adjacent nodes, $i$ and $j$,
is eventually monotone increasing or monotone decreasing. That is, there is a time $T$ such that exactly one of the inequalities $c_{ij}(t_1) \leq c_{ij}(t_2)$ and $c_{ij}(t_1) \geq c_{ij}(t_2)$ holds for all times $t_2 > t_1 \geq T$.

Furthermore, if nodes $i$ and $j$ are in different limit opinion clusters, then $c_{ij}(t)$ is eventually monotone decreasing.  

\end{lemma}

\begin{lemma}\label{lemma:convergence_DW}
Consider our adaptive-confidence DW model (with update rules \cref{eq:DW_update_rule} and \cref{eq:HK_confidence}) with parameters $\gamma \in (0,1]$ and $\delta \in [0,1)$.
Suppose that $c_{ij} = \lim\limits_{t\to\infty} c_{ij}(t)$ exists. It then follows that, almost surely, either $c_{ij} = 0$ or $c_{ij} = 1$.
\end{lemma}

% -------------------------------------------
\subsubsection{Effective-graph analysis}
We now present \cref{thm:effgraph_DW}, which is our main result about effective graphs in our adaptive-confidence DW model.  
In \cref{appendix:proof_eff_DW}, we present its proof, which uses results from \cref{sec:confidence_bound_analysis_DW}.

\begin{thm}\label{thm:effgraph_DW}
In our adaptive-confidence DW model (with update rules \cref{eq:DW_update_rule} and \cref{eq:HK_confidence}) with parameters $\gamma \in (0,1]$ and $\delta \in [0,1)$, the
effective graph $G_\mathrm{eff}(t)$ almost surely eventually
has edges only between nodes of the same limit opinion cluster.
That is, there is almost surely some time $T$ such that 
for all times $t \geq T$, we have that $(i,j) \in E_\mathrm{eff}(t)$ implies that nodes $i$ and $j$ are in the same limit opinion cluster (i.e., $x^i = x^j$).
\end{thm}

Unlike in our adaptive-confidence HK model, $\lim\limits_{t\to\infty} G_\mathrm{eff}(t)$ may not exist in our adaptive-confidence DW model. When the limit does exist, we refer to $\lim\limits_{t\to\infty} G_\mathrm{eff}(t)$ 
as the \emph{limit effective graph}.

For completeness, we now state \cref{thm:effgraph_baseline_DW}, which guarantees the {almost-sure} convergence of the effective graphs as $t \rightarrow \infty$ in the baseline DW model.
In \cref{appendix:proof_eff_baseline_DW}, we prove \cref{thm:effgraph_baseline_DW} by first proving 
\cref{lemma:baseDW_eff_edges} and \cref{lemma:baseDW_less_than_c}.
Because the baseline DW model has a time-independent confidence bound $c$, our proof of \cref{thm:effgraph_baseline_DW} uses different ideas than our proof of \cref{thm:effgraph_DW}.

\begin{thm}\label{thm:effgraph_baseline_DW}
Consider the baseline DW model (with update rule \cref{eq:DW_update_rule}).
Almost surely, the effective graph $G_\mathrm{eff}(t)$ is eventually constant with respect to time. 
That is, there is almost surely a time $T$ such that $G_\mathrm{eff}(t) = G_\mathrm{eff}(T)$ for all times $t \geq T$. 
\\ \indent
Furthermore, suppose that the limit effective graph $\lim\limits_{t\to\infty} G_\mathrm{eff}(t)$ exists.
If adjacent nodes $i$ and $j$ have the same limit opinion (i.e., if $x^i = x^j$), then the edge $(i,j)$ is in the limit effective graph. 
Additionally, if the edge $(i,j)$ is in the limit effective graph, then $x^i = x^j$ almost surely.
\end{thm}

% -------------------------------------------
% ----------------------------------------------
\section{Details of our numerical simulations}\label{sec:simulations}

We now discuss the details of our numerical simulations of our adaptive-confidence HK and DW models.

% -------------------------------------------
\subsection{Network structures}\label{sec:simulations_networks}

We first simulate our adaptive-confidence HK and DW models on complete graphs, and we then examine our models on a variety of networks and study how different network structures affect their behaviors.
We simulate our adaptive-confidence HK model on synthetic networks that we generate using random-graph models, and we simulate both adaptive-confidence BCMs on networks from empirical data. Because of computational limitations, we consider larger networks for the adaptive-confidence HK model than for the adaptive-model DW model.

We simulate our adaptive-confidence HK model on a complete graph, $G(N,p)$ Erd\H{o}s--R\'enyi (ER) random graphs, and two-community stochastic-block-model (SBM) random graphs. In each case, we consider graphs with 1000 nodes. We also simulate our adaptive-confidence HK model on social networks from the {\sc Facebook100} data set \cite{red2011, Facebook100}.

A $G(N,p)$ ER graph has $N$ nodes and independent probability $p$ of an edge between each pair of distinct nodes \cite{newman2018}. When $p = 1$, this yields a complete graph. We consider $G(N, p)$ graphs with $p \in \{0.1, 0.5\}$ to vary the sparsity of the graphs while still yielding connected graphs for our simulations. 
All of the ER graphs in our simulations are connected.

To determine how a network with an underlying community structure affects the dynamics of our adaptive-confidence HK model, we consider undirected SBM networks \cite{newman2018} with a $2 \times 2$ block structure in which each block corresponds to an ER graph. To construct these SBMs, we partition a network into two sets of nodes; one set (which we denote by \textrm{A}) has 75\% of the nodes, and the other set (which we denote by \textrm{B}) has the remaining 25\% of the nodes. Our choice is inspired by the two-community SBM that was considered in \cite{kureh2020}.
We define a symmetric edge-probability matrix
\begin{equation}
    P = \begin{bmatrix}
    P_\textrm{AA} & P_\textrm{AB} \\ P_\textrm{AB} & P_\textrm{BB}
    \end{bmatrix} \,,
\end{equation}
where $P_\textrm{AA}$ and $P_\textrm{BB}$ are the probabilities of an edge between two nodes within the sets \textrm{A} and \textrm{B}, respectively, and $P_\textrm{AB}$ is the probability of an edge between a node in set \textrm{A} and a node in set \textrm{B}. In our simulations, $P_\textrm{AA} = P_\textrm{BB} = 1$ and $P_\textrm{AB} = 0.01$. 

In addition to synthetic networks, we also simulate our adaptive-confidence HK model on several real-world networks. For each network, we use the largest connected component\footnote{A connected component \cite{newman2018} of an undirected network $G$ is a maximal subgraph with a path between each pair of nodes.} 
(LCC). In \cref{tab:network_sizes}, we give the numbers of nodes and edges in the LCCs of these networks, which are social networks from the 
{\sc Facebook100} data set \cite{red2011, Facebook100}. 
In each of these networks, the nodes are the Facebook pages of individuals at a university and the edges encode Facebook ``friendships'' between individuals in a one-day snapshot of the network from fall 2005 \cite{red2011, Facebook100}.
The numbers of nodes in the LCCs of the examined Facebook networks range from 962 to 14,917.

For our adaptive-confidence DW model, we examine a complete graph and one real-world network.
We simulate our adaptive-confidence DW model on a 100-node complete graph, which is one tenth of the size of the complete graph that we consider for our adaptive-confidence HK model. 
We use this smaller size because of computational limitations. 
Our simulations of our adaptive-confidence DW model on a 100-node complete graph frequently reach our ``bailout time'' (see \cref{sec:simulations_specs} and \cref{tab:DW_bailout}) for small initial confidence bounds.
We also simulate our adaptive-confidence DW model on the LCC of the real-world {\sc NetScience} network of coauthorships between researchers in network science \cite{netscience}.

\begin{table}[ht]
\centering
\caption{\label{tab:network_sizes} The real-world networks on which we simulate our adaptive-confidence BCMs. For each network, we use the largest connected component and indicate the numbers of nodes and edges in that component.}
\begin{tabular}{l | l | l | l}
\hline\hline
Network    & Number of Nodes & Number of Edges & Model \\ \hline
{\sc NetScience} & 379            & 914             & DW              \\
{\sc Reed}       & 962            & 18,812          & HK              \\
{\sc Swarthmore} & 1657           & 61,049          & HK              \\
{\sc Oberlin}    & 2920           & 89,912          & HK              \\
{\sc Pepperdine} & 3440           & 152,003         & HK              \\
{\sc Rice}       & 4083           & 184,826         & HK              \\
{\sc UC Santa Barbara}       & 14,917          & 482,215         & HK              \\ 
\hline\hline
\end{tabular}
\end{table}

 % -------------------------------------------
\subsection{Simulation specifications}\label{sec:simulations_specs}

In \cref{tab:parameters}, we indicate the values of the model parameters that we examine in our simulations of our BCMs. 
The {BCM} parameters are
the confidence-increase parameter $\gamma$, the confidence-decrease parameter $\delta$, the initial confidence bound $c_0$,
and (for the adaptive-confidence DW model only) the compromise parameter $\mu$.
For both our HK and DW models, the parameter pair $(\gamma,\delta) = (0,1)$ corresponds to the associated baseline BCM.

Our BCM simulations include randomness both from the initial opinions of the nodes and from generating specific networks in random-graph ensembles.
The adaptive-confidence DW model also has randomness from the selection of nodes at each time step.
We use Monte Carlo simulations to mitigate the effects of noise.
For each parameter set of a random-graph model (i.e., the ER and SBM graphs), we generate 5 graphs.
Additionally, for each graph, we generate 10 sets of initial opinions uniformly at random and reuse these sets of opinions for all BCM parameter values.

\begin{table}[ht]
\centering
\caption{\label{tab:parameters} The {BCM} parameter values that we examine in simulations of our adaptive-confidence BCMs.
We consider more parameter values for complete graphs than for the other networks. We consider all of the indicated values for complete graphs, and we consider values without the asterisk ($^*$) for the ER, SBM, and real-world networks.
}
\begin{tabular}{ll}
\hline\hline
\textbf{Model} & \textbf{{BCM} Parameters} \\ \hline
Adaptive-Confidence HK &
{$\!\begin{aligned}
\gamma &\in \{0, 0.0001^*, 0.0005^*, 0.001, 0.005, 0.01, 0.05, 0.1^*\}\\ 
    \delta &\in \{0.01^*, 0.1^*, 0.5, 0.9, 0.95, 0.99, 1\}\\
c_0 &\in \{0.02, 0.03, \dots, 0.19, 0.20, 0.30, 0.40, 0.50\}\\ 
\end{aligned}$} \\
\hline
Adaptive-Confidence DW & 
{$\!\begin{aligned}
 \gamma &\in \{0.1, 0.3, 0.5^*\}\\ 
	\delta &\in \{0.3^*, 0.5, 0.7^*\}\\
c_0 &\in \{0.1, 0.2, 0.3, 0.4, 0.5, 0.6, 0.7, 0.8, 0.9\}\\ 
	\mu &\in \{0.1, 0.3, 0.5\}
\end{aligned}$} \\
\hline\hline
\multicolumn{2}{l}{
\footnotesize{$^*$ We consider these parameter values only for complete graphs.}
}
\end{tabular}
\end{table}

For our numerical simulations, we need a stopping criterion, as it can potentially take arbitrarily long for nodes to reach their limit opinions in a BCM simulation. We consider the effective graph $G_\mathrm{eff}(t)$, which we recall (see \cref{sec:theory}) is the subgraph of the original graph $G$ with edges only between nodes that are receptive to
each other at time $t$. In our simulations, each of the connected components of $G_\mathrm{eff}(t)$ is an ``opinion cluster'' $K_r(t)$ at time $t$. Our stopping criterion checks that the maximum difference in opinions between nodes in the same opinion cluster is less than some tolerance. That is,
\begin{equation}  \label{eq:stopping_criterion}
    \max \{ |x_i(t) - x_j(t)| \text{ such that } i,j \in K_r(t) \text{ for some } r\} < \text{tolerance}\,.
\end{equation}
We use a tolerance value of $1 \times 10^{-6}$ for our adaptive-confidence HK model. Because of computational limitations, we use a tolerance value of $0.02$ for our adaptive-confidence DW model.
We refer to the time $T_f$ at which we reach our stopping criterion as the ``convergence time'' of our simulations.
In our simulations, the ``final effective graph'' is the effective graph at the convergence time (i.e., the time $T_f$ that a simulation reaches our stopping criterion). 
We refer to the connected components of the final effective graph as the ``final opinion clusters'' of a simulation; they approximate the limit opinion clusters.

Our theoretical results about effective graphs inform our stopping criterion.
\cref{thm:effgraph_HK} and \cref{thm:effgraph_baseline_HK} give theoretical guarantees for our adaptive-confidence HK model and the baseline HK model, respectively, that eventually the only edges of an effective graph are those between adjacent nodes in the same limit opinion cluster.
\cref{thm:effgraph_DW} and \cref{thm:effgraph_baseline_DW}
give similar but weaker guarantees for our adaptive-confidence DW model and the baseline DW model.
Consequently, if one of our simulations runs for sufficiently many time steps, its final opinion clusters are a good approximation of the limit opinion clusters. 
However, instead of imposing a set number of time steps for our simulations, we use a tolerance value as a proxy to determine a ``sufficient'' number of time steps.

The final and limit opinion clusters {in our models} may not be the same, as our choice of tolerance values can lead to simulations stopping before we can determine their limit opinion clusters. 
Additionally, if two distinct connected components in an effective graph converge to the same limit opinion value, 
the nodes in those connected components are in the same limit opinion cluster.
However, the sets of nodes that constitute these connected components are distinct final opinion clusters.
In practice, our simulations are unlikely to have distinct opinion clusters that converge to the same opinion in the infinite-time limit. Therefore, for small tolerance values, our final opinion clusters are a good approximation of the limit opinion clusters. 

To ensure that our simulations stop after a reasonable amount of time, we use
a bailout time of $10^6$ time steps. Our simulations of the adaptive-confidence HK model in the present paper never reach this bailout time. 
However, our simulations of the adaptive-confidence DW model frequently reach the bailout time for small values of $c_0$.
See \cref{sec:DW_complete} and \cref{tab:DW_bailout}.

% ----------------------------------------------
\subsection{Quantifying model behaviors}\label{sec:simulations_quantities}
In our numerical simulations, we investigate the convergence time and characterize the final opinions.
To examine the convergence time, we record the number $T_f$ of time steps that it takes for simulations to reach our stopping criterion.
To characterize opinions, we determine whether there is consensus or opinion fragmentation (which we define shortly), quantify the opinion fragmentation using Shannon entropy, and examine the numbers of nodes and edges in each opinion cluster.

In models of opinion dynamics, it is common to investigate whether or not nodes eventually reach a consensus (i.e., arrive at the same opinion) \cite{noorazar-review}.
In practice, to determine whether a simulation reaches a consensus state,
we use notions of ``major'' and ``minor'' clusters.
Consider a 1000-node network in which 998 nodes have one steady-state opinion but the remaining 2 nodes have a different steady-state opinion. 
In applications, it does not seem appropriate to characterize this situation as opinion polarization or fragmentation. 
Therefore, we use notions of major and minor opinion clusters \cite{lorenz2008, laguna2004}, which we characterize in an ad hoc way. 
We define a ``major'' opinion cluster as a final opinion cluster with strictly more than 1\% of the nodes of a network. A final opinion cluster that is not a major cluster is a ``minor'' cluster.
We say that a simulation that results in one major cluster yields a ``consensus'' state and that a simulation that results in at least two major clusters yields ``opinion fragmentation'' (i.e., a ``fragmented'' state).\footnote{In studies of opinion dynamics, it is common to use the term ``fragmentation'' to refer to situations with three or more opinion clusters and to use the term ``polarization'' to refer to situations with exactly two opinion clusters. In the present paper, it is convenient to quantify any state other than consensus as a fragmented state.}  
We track the numbers and sizes of all major and minor clusters, and we use all clusters (i.e., both major and minor clusters) to quantify opinion fragmentation.

There are many ways to quantify opinion fragmentation \cite{bramson2016, musco2021, adams2022}.
We distinguish situations in which the final opinion clusters (major or minor) are of similar sizes from situations in which these clusters have a broad range of sizes. 
Following Han et al.~\cite{han2020}, we calculate Shannon entropy to quantify opinion fragmentation.\footnote{See~\cite{grace2022} for another study that followed \cite{han2020} by calculating Shannon entropy to help quantify opinion fragmentation in a BCM.}
At some time $t$, suppose that there are $R$ opinion clusters, which we denote by $K_r(t)$ for $r \in \{1, 2, \ldots, R\}$. We refer to the set $\{K_r(t)\}_{r = 1}^R$, which is a partition of the set of nodes of a network, as an ``opinion-cluster profile''. The fraction of nodes in opinion cluster $K_r(t)$ is $|K_r(t)|/N$, where $N$ is the number of nodes of a network. The Shannon entropy $H(t)$ of an opinion-cluster profile is
\begin{equation} \label{eq:entropy}
    H(t) = - \sum\limits_{r=1}^R \frac{|K_r(t)|}{N} \ln \left( \frac{|K_r(t)|}{N} \right) \, .
\end{equation}
We calculate $H(T_f)$, which is the Shannon entropy of the opinion-cluster profile of final opinion clusters at convergence time. Computing Shannon entropy allows us to use a scalar value to quantify the distribution of opinion-cluster sizes, with larger entropies indicating greater opinion fragmentation. The Shannon entropy is larger when there are more opinion clusters. Additionally, for a fixed number $R$ of opinion clusters, the Shannon entropy is larger when the opinion clusters are evenly sized than when the sizes are heterogeneous. When comparing two opinion-cluster profiles, we consider both the numbers of major clusters and the Shannon entropies. When there are sufficiently few minor clusters, we expect the number of major clusters to follow the same trend as the Shannon entropy.

To examine the structure of final opinion clusters, we study the properties of final effective graphs.
Our theoretical results allow the possibility that some adjacent nodes 
converge to the same opinion without being mutually receptive.
We observe this phenomenon in our numerical simulations.
(See \cref{sec:results} for more discussion.)
To quantify this behavior, we calculate a weighted average of the fractions of edges (which we call the ``weighted-average edge fraction'') that are in each opinion cluster of the final effective graph.
In an opinion-cluster profile $\{K_i(t)\}_{r = 1}^R$, let $E(r)$ denote the set of edges of the original graph $G$ between nodes in opinion cluster $r$ and let $E_\mathrm{eff}(t,r)$ denote the set of edges of the effective graph $G_\mathrm{eff}(t)$ that are in opinion cluster $r$. 
That is,
\begin{align*}
	E(r) &= \{(i,j) \in E \text{ such that } i,j \in K_r(t)\}\,, \\
	E_\mathrm{eff}(t, r) &= \{(i,j) \in E_\mathrm{eff}(t) \text{ such that } i,j \in K_r(t)\} \, .
\end{align*}
The weighted average of the fractions of edges (i.e., the weighted-average edge fraction) that are in the effective graph for each opinion cluster is
\begin{equation} \label{eq:weighted_avg}
    W(t) = \sum\limits_{\substack{r = 1\\E(r)\neq 0} }^R 
    \left(\frac{|K_r(t)|}{N - \ell} \right)
    \left(\frac{|E_\mathrm{eff}(t, r)|}{|E(r)|} \right)\,,
\end{equation}
where $\ell$ is the number of isolated nodes of the effective graph.\footnote{If every component of an effective graph is an isolated node (i.e, $N = \ell$), then one can take either $W(t) = 0$ or $W(t) = 1$.   
In our simulations, this situation never occurred.}
An isolated node is an opinion cluster with $E(r) = 0$.
We are interested in the value of $W(t)$ at the convergence time $T_f$.
Therefore, we calculate the weighted-average edge fraction $W(T_f)$.
If each opinion cluster of an effective graph has all of its associated original edges of $G$, then $W = 1$. 
The value of $W$ is progressively smaller when there are progressively fewer edges between nodes 
in the same opinion cluster of the effective graph. In \cref{fig:G_eff}, we show examples of effective graphs with {$W(T_f) < 1$}.

\begin{figure}
    \centering
    \subfloat[One final opinion cluster in a simulation of our adaptive-confidence HK model that does not reach consensus on a 1000-node complete graph with $\gamma = 0.001$, $\delta = 0.5$, and $c_0 = 0.1$.]{\includegraphics[width=0.4\textwidth]{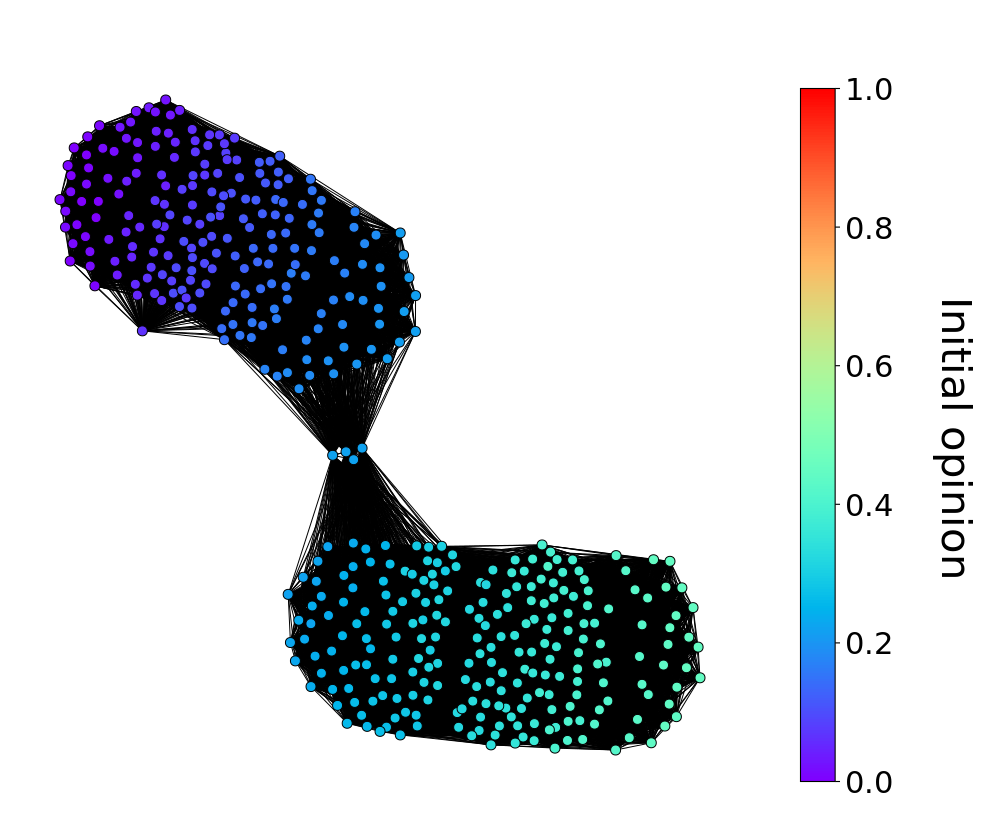}}
    \hspace{1cm}
    \subfloat[The final effective graph in a simulation of our adaptive-confidence DW model that reaches consensus on a 100-node complete graph with $\gamma = 0.1$, $\delta = 0.5$, $c_0 = 0.1$, and $\mu = 0.1$.]{\includegraphics[width=0.4\columnwidth]{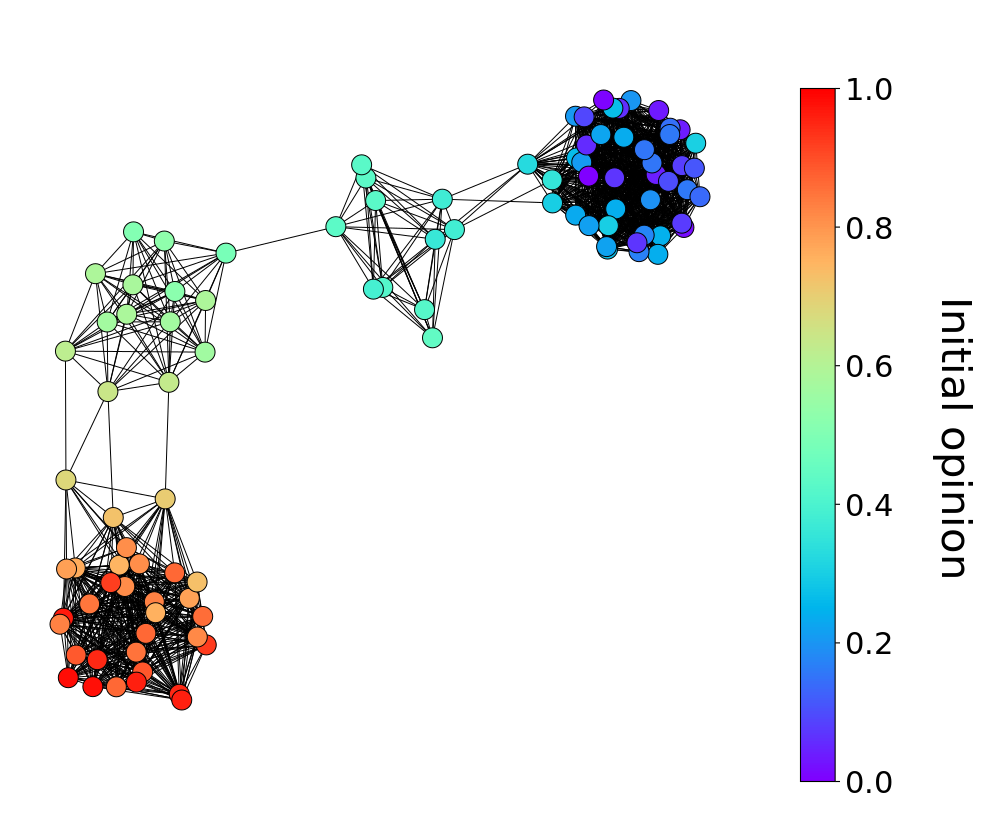}}
    \caption{Examples of final effective graphs with $W(T_f) < 1$. We color the nodes by their initial opinion values.}
    \label{fig:G_eff}
\end{figure}

% -------------------------------------------
% -------------------------------------------
\section{Results of our numerical simulations}\label{sec:results}
We now present the results of numerical simulations of our adaptive-confidence BCMs. 
We consider various values of the BCM parameters, which are the confidence-increase parameter $\gamma$, the confidence-decrease parameter $\delta$, the initial confidence bound $c_0$, and (for the adaptive-confidence DW model only) the compromise parameter $\mu$.
We use the BCM-parameter values in \cref{tab:parameters}, including the values that correspond to the baseline models (i.e., $(\gamma,\delta) = (0,1)$).
Our code and plots are available in our \href{https://gitlab.com/graceli1/Adaptive-Confidence-BCM}{code repository}.

As we described in \cref{sec:simulations_quantities}, for both of our adaptive-confidence BCMs, we examine 
the number of major clusters (which we use to determine whether a simulation reaches a consensus state or a fragmented state),
the number of minor clusters,
the Shannon entropy $H(T_f)$ (see equation \cref{eq:entropy}), 
the weighted-average edge fraction $W(T_f)$ (see equation \cref{eq:weighted_avg}), and the convergence time $T_f$.
When the Shannon entropy and the number of major clusters follow similar trends, we only show results for the number of major clusters, as it is easier to interpret than the entropy. 
{To avoid drowning readers with too much repetition, we include some of our plots of our adaptive-confidence HK model and adaptive-confidence DW model in \cref{appendix:HK_simulations} and \cref{appendix:DW_simulations}, respectively.
Furthermore, we do not show plots for all of our numerical results; the omitted plots are available in our \href{https://gitlab.com/graceli1/Adaptive-Confidence-BCM}{code repository}.}

Our simulation results and theoretical results about effective graphs complement each other.
In \cref{thm:convergence_HK}, we proved for our adaptive-confidence HK model that all dyadic confidence bounds converge either to $0$ or to $1$. 
We also proved that the dyadic confidence bounds for node pairs in different limit opinion clusters must converge to $0$.
However, we have not proven
whether or not the dyadic confidence bounds for nodes in the same limit opinion cluster converge to $1$, 
so it is possible for such confidence bounds to converge to $0$. 
(We first mentioned this point in \cref{sec:conf_bound_analysis_HK}.)
If a dyadic confidence bound convergences to $0$, then the corresponding edge is absent in the limit effective graph (which is guaranteed to exist by \cref{thm:effgraph_HK}).
Our numerical simulations suggest that a final opinion cluster can include adjacent nodes  
whose dyadic confidence bound converges to $0$.
In particular, in many simulations, we observe that the 
weighted-average edge fraction $W(T_f) < 1$, which corresponds to absent edges of the final effective graph between nodes that are in the same final opinion cluster.
For our adaptive-confidence DW model, we prove analogous theoretical results (see \cref{thm:convergence_DW} and \cref{thm:effgraph_DW}) and we again observe simulations with $W(T_f) < 1$.

%%%%%%%

\subsection{Adaptive-confidence HK model}\label{sec:HK_results}
\subsubsection{Summary of our simulation results}\label{sec:HK_results_summary}

\newcolumntype{t}{X}
\newcolumntype{q}{>{\hsize=0.27\hsize}X}
\renewcommand\tabularxcolumn[1]{m{#1}} %Vertical centering of tabularx columns

\begin{table}[!ht]
\centering
\caption{\label{tab:HK_trends} Summary of the observed trends in our adaptive-confidence HK model. Unless we note otherwise, we observe these trends for the complete graph, all examined random-graph models, and all examined real-world networks.
}
\noindent
\begin{tabularx}{\textwidth}{qt}
\hline\hline
Quantity & Trends \\\hline
Convergence Time
&
$\bullet$ For fixed values of the initial confidence bound $c_0$,
our adaptive-confidence HK model tends to converge more slowly than the baseline HK model.
\smallskip \newline 
$\bullet$ When our simulations reach a consensus state, for fixed values of $c_0$ and the confidence-increase parameter $\gamma$, our model converges faster when
the confidence-decrease parameter $\delta = 1$ 
than when $\delta \leq 0.9$.
For $\delta \in \{0.95, 0.99\}$, the convergence time transitions from the $\delta \leq 0.9$ behavior to the $\delta = 1$ behavior as we increase $c_0$.
\\ \hline
Opinion \newline Fragmentation
& 
$\bullet$ Our adaptive-confidence HK model yields consensus for $\gamma \geq 0.05$. \smallskip \newline
$\bullet$ For fixed values of $c_0$,
our adaptive-confidence HK model tends to yield 
fewer major clusters than the baseline HK model. 
When we fix the other BCM parameters, the number of major clusters tends to decrease  
as either (1) we decrease $\delta$ or (2) we increase $\gamma$.
\smallskip \newline
$\bullet$ For our synthetic networks, we observe that the trends in Shannon entropy match the trends in the numbers of major clusters and that our adaptive-confidence HK model tends to yield less opinion fragmentation than the baseline HK model.$^*$
\smallskip \newline
$\bullet$ 
For the baseline HK model, as we increase $c_0$, the number of major clusters tends to decrease. In our adaptive-confidence HK model, for simulations without consensus and for sufficiently large $\gamma$,
the number of major clusters tends to first increase and then decrease as we increase $c_0$. 
\\ \hline
$W(T_f)$ 
& 
$\bullet$ When $\delta = 1$, both our adaptive-confidence HK model and the baseline HK model yield $W(T_f) = 1$. \smallskip \newline
$\bullet$ For fixed $\gamma$ and $c_0$, as we increase $\delta$, the weighted-average edge fraction $W(T_f)$ tends to decrease. \smallskip \newline
$\bullet$ For simulations that reach a consensus state, we observe two qualitative behaviors for $W(T_f)$: for $\delta \leq 0.9$, we observe that $W(T_f) <  1$ and that there is a seemingly linear relationship between $W(T_f)$ and $c_0$; for $\delta = 1$, we observe that $W(T_f) = 1$.
Additionally, for $\delta \in \{0.95, 0.99\}$, the behavior of $W(T_f)$ transitions from the $\delta \leq 0.9$ behavior to the $\delta = 1$ behavior as we increase $c_0$.\\
\hline\hline
\multicolumn{2}{>{\hsize=1.4\hsize}X}{
\footnotesize{$^*$ For the {\sc Facebook100} networks, we do not observe this trend, seemingly because of the large numbers of minor clusters (which are incorporated in our calculation of Shannon entropy in \cref{eq:entropy}) for these networks.}
}
\end{tabularx}
\end{table}

For our adaptive-confidence HK model, all of our numerical simulations reach
a consensus state for $c_0 \geq 0.3$.
(As we discussed in \cref{sec:simulations_quantities}, a consensus state has exactly one major opinion cluster.)
We show our simulation results for $c_0 \in \{0.02, 0.03, \ldots, 0.20 \}$; we include the results for the other
examined values of $c_0$ 
(see \cref{tab:parameters}) in our \href{https://gitlab.com/graceli1/Adaptive-Confidence-BCM}{code repository}.
We examine the numbers of major and minor clusters, the Shannon entropy $H(T_f)$ (see equation \cref{eq:entropy}), the weighted-average edge fraction $W(T_f)$ (see equation \cref{eq:weighted_avg}), and the convergence time $T_f$. 
We plot each of these quantities versus the initial confidence bound $c_0$.
For each value of the confidence-increase parameter $\gamma$, we generate one plot; each plot has one curve for each value of the confidence-decrease parameter $\delta$.
Each point in our plots is a mean of our numerical simulations for the associated values of the BCM parameter set ($\gamma$, $\delta$, and $c_0$).
We also show one standard deviation from the mean.
For our simulations on a complete graph and the {\sc Facebook100} networks, each point in our plots is a mean of 10 simulations (from 10 sets of initial opinions). For our simulations on $G(N,p)$ ER random graphs and SBM random graphs, each point in our plots is a mean of 50 simulations (from 5 random graphs that each have 10 sets of initial opinions).
We include all plots --- including those that we do not present in the present section (i.e., \cref{sec:HK_results}) or in \cref{appendix:HK_simulations} --- in our 
\href{https://gitlab.com/graceli1/Adaptive-Confidence-BCM}{code repository}.
In \cref{tab:HK_trends}, we summarize the trends that we observe in our simulations.

In all of our simulations of our adaptive-confidence HK model, we observe that $\gamma \geq 0.001$ results in fewer major clusters than in the baseline HK model. 
For a fixed initial confidence bound $c_0$, our adaptive-confidence HK model tends to yield fewer major opinion clusters and less opinion fragmentation as either (1) we increase $\gamma$ for fixed $\delta$ or (2) we decrease $\delta$ for fixed $\gamma$.
Intuitively, one expects larger values of $\gamma$ to encourage consensus because a larger $\gamma$ entails a larger increase in a dyadic confidence bound after a positive interaction.
Less intuitively, smaller values of $\delta$, which entail a larger decrease in a dyadic confidence bound after a negative interaction, also seem to encourage consensus.
In our adaptive-confidence HK model, we update opinions synchronously, with each node interacting with of all its neighboring nodes at each time. 
When two adjacent nodes are mutually unreceptive, their dyadic confidence bound decreases.
Given the synchronous updates in our adaptive-confidence HK model, we hypothesize that small values of $\delta$ yield
a faster decrease than large values of $\delta$ in the dyadic confidence bound between mutually unreceptive nodes.
For small values of $\delta$, pairs of nodes may quickly become mutually unreceptive and remain mutually unreceptive.
Meanwhile, individual nodes can be receptive to (and thus average) fewer ``conflicting'' opinions\footnote{When the neighbors to which a node is receptive have large differences in opinions with each other, we say that that node is receptive to ``conflicting'' opinions.\label{footnote:conflicting_opinion}}, possibly aiding in reaching a consensus.

In the baseline HK model, the number of major opinion clusters tends to decrease as we increase $c_0$.
For intermediate values of $\gamma$ (e.g., $\gamma \in \{0.005, 0.01\}$ for a complete graph; see panels (E) and (F) in \cref{fig:HK_complete_clusters}),
we do not observe this trend in our adaptive-confidence HK model. 
Instead, as we increase $c_0$, we observe an increase and then a decrease in the number of major clusters; 
unlike for the baseline HK model, small values of $c_0$ tend to promote more consensus (i.e., it tends to yield fewer major clusters).
For smaller values of $c_0$, it seems that nodes tend to be receptive to fewer nodes early in a simulation, so fewer opinions can influence them.
Therefore, for these values of $c_0$,
nodes that are mutually receptive may quickly approach a consensus when $\gamma$ is sufficiently large.
Our calculations of the weighted-average edge fraction $W(T_f)$ support this hypothesis.
For sufficiently large $\gamma$, small values of $c_0$ yield small values of $W(T_f)$, indicating that 
final opinion clusters tend to have many pairs of mutually unreceptive nodes.

% ----------------------------------------
\subsubsection{A complete graph}\label{sec:HK_complete}
We now discuss the simulations of our adaptive-confidence HK model on a complete graph. 
We summarize the observed trends in \cref{tab:HK_trends}.

\begin{figure} [ht]
\label{fig:HK_complete_clusters}
\includegraphics[width=\columnwidth]{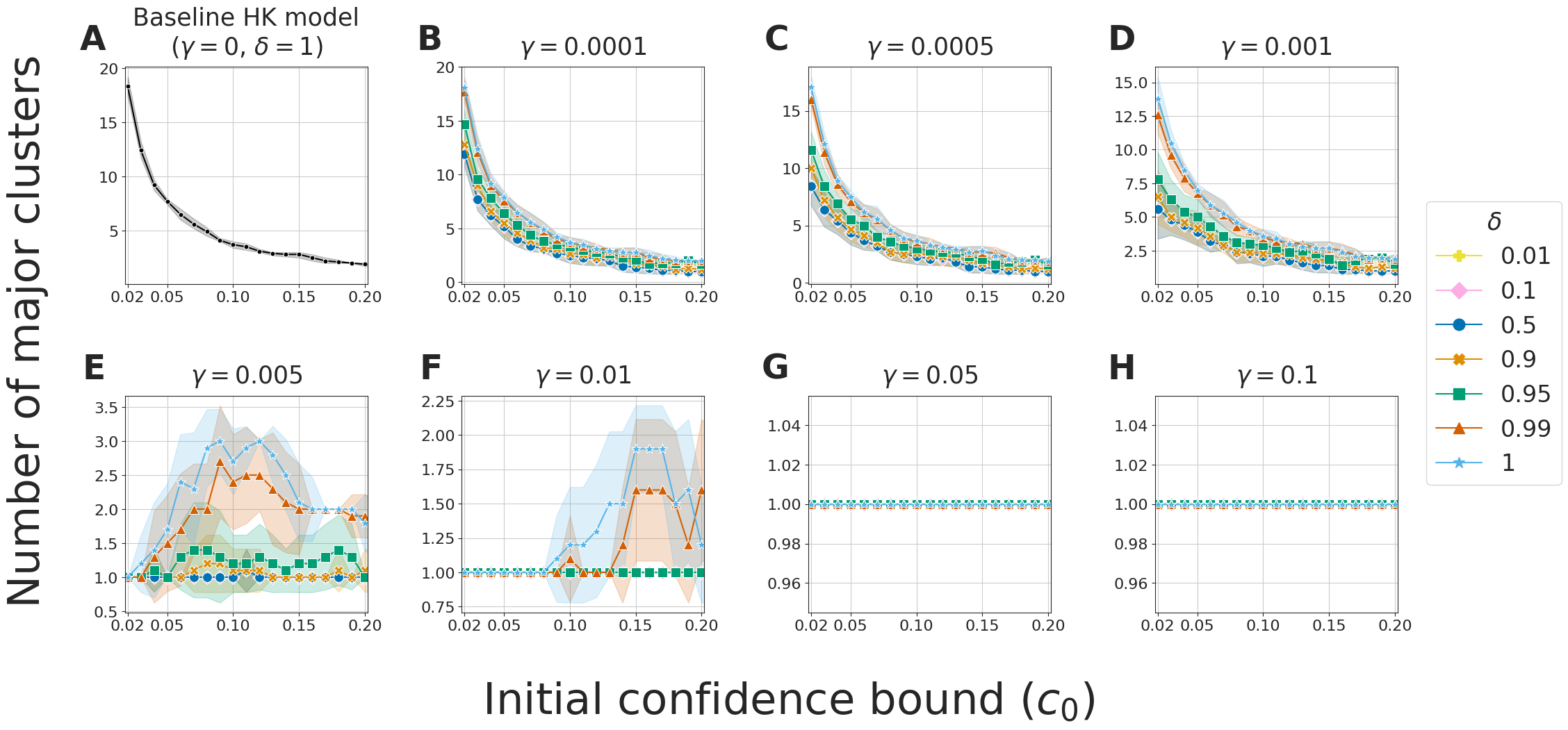}
\caption{The numbers of major clusters in simulations of our adaptive-confidence HK model on a 1000-node complete graph for various combinations of the BCM parameters $\gamma$, $\delta$, and $c_0$. 
In this and subsequent figures, we plot the mean value of our simulations for each set of BCM parameters. The bands around each curve indicate one standard deviation around the mean values.
For clarity, in this figure and in subsequent figures, the vertical axes of different panels have different scales.}
\end{figure}

In \cref{fig:HK_complete_clusters}, we observe for a 1000-node complete graph that our adaptive-confidence HK model yields fewer major clusters (i.e., it encourages more consensus) than the baseline HK model for a wide range of BCM parameter values. Our adaptive-confidence HK model always reaches a consensus state for $\gamma \geq 0.05$. In our simulations that do not reach consensus, we tend to observe progressively more major clusters and more opinion fragmentation as either (1) we decrease $\gamma$ for fixed $\delta$ and $c_0$ or (2) we increase $\delta$ for fixed $\gamma$ and $c_0$. 
For the baseline HK model and our adaptive-confidence HK model with small values of $\gamma$ (specifically, $\gamma \in \{0.0001, 0.0005, 0.001\}$), the number of major opinion clusters tends to decrease as we increase $c_0$.
We do not observe this trend for larger values of $\gamma$ (specifically, $\gamma \in \{0.005, 0.01\}$). 
Instead, for these values of $\gamma$, small values of $c_0$ tend to promote more consensus. For example, simulations always reach a consensus state when $\gamma = 0.01$ and $c_0 \leq 0.08$.
At the end of \cref{sec:HK_results_summary}, we 
discussed our hypothesis behind this observation.

We observe very few minor clusters in our simulations of our adaptive-confidence HK model on a 1000-node complete graph.
For each value of the BCM parameter set $(\gamma, \delta, c_0)$, the mean number of minor clusters in our 10 simulations is bounded above by $1$.
Because there are few minor clusters, the Shannon entropy (which accounts for both major clusters and minor clusters; see equation \cref{eq:entropy}) and the number of major clusters follow similar trends for a 1000-node complete graph.

\begin{figure} [ht]
\label{fig:HK_complete_WT}
\includegraphics[width=\columnwidth]{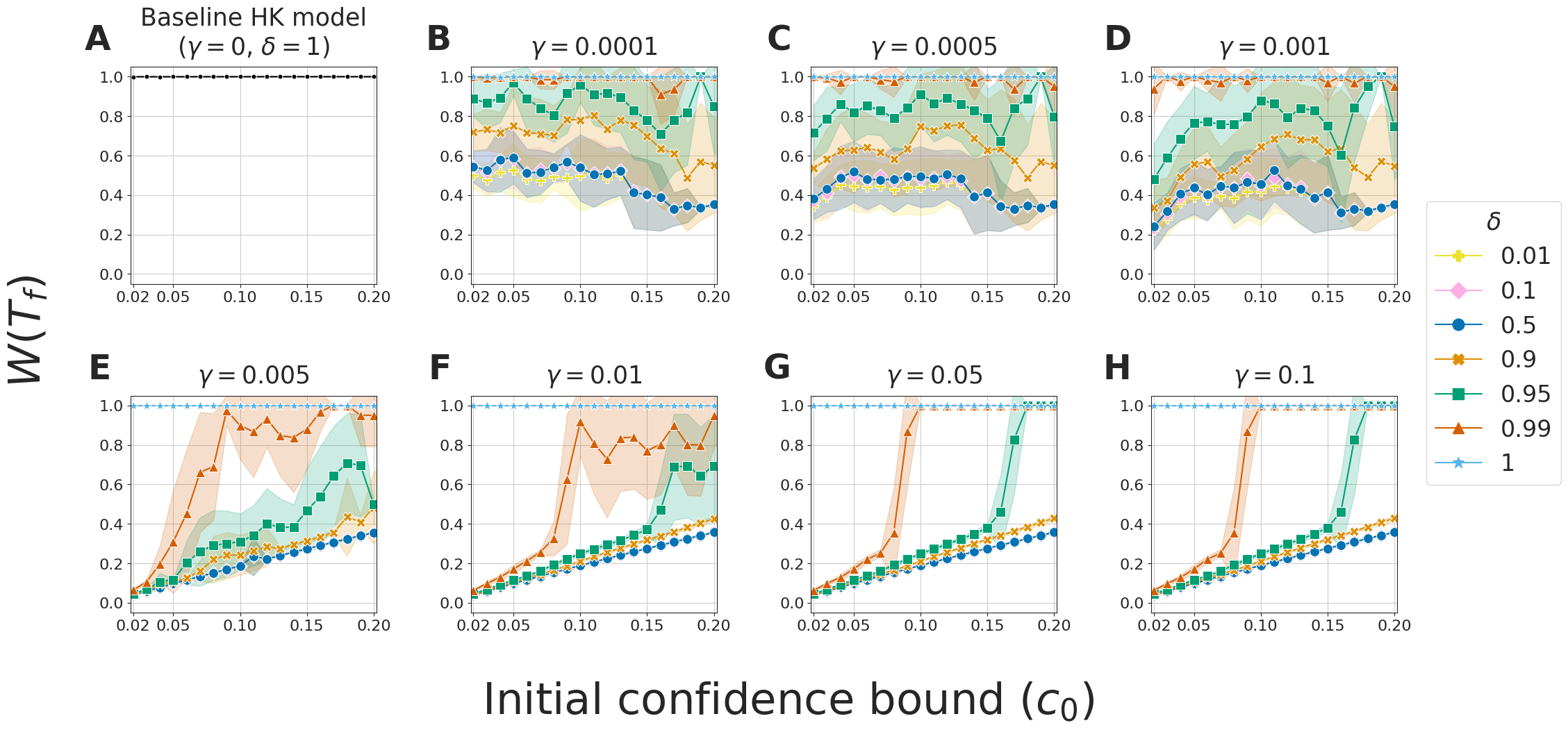}
\caption{
The weighted-average edge fraction $W(T_f)$ (see equation \cref{eq:weighted_avg})
in simulations of our adaptive-confidence HK model on a 1000-node complete graph for various combinations of the BCM parameters $\gamma$, $\delta$, and $c_0$.
}
\end{figure}

In \cref{fig:HK_complete_WT}, we show $W(T_f)$
(see equation \cref{eq:weighted_avg}), which is the weighted average of the fractions of edges in the opinion clusters of the final effective graph.
When $\delta = 1$, both our adaptive-confidence HK model and the baseline HK model yield $W(T_f) = 1$. This indicates that all final opinion clusters (i.e., the connected components of the effective graph at time $T_f$) are complete subgraphs (i.e., cliques).
By contrast, in our adaptive-confidence HK model, when $\delta < 1$ and for a wide range of the other BCM parameters, we observe that $W(T_f) < 1$. This indicates that some nodes that are adjacent in the graph $G$ and in the same final opinion cluster do not have an edge between them in the final effective graph $G_\mathrm{eff}(T_f)$. 
The nodes in these dyads are thus not receptive to each other (and hence do not influence each other's opinions),
but they nevertheless converge to the same opinion. 
When $\delta$ is sufficiently small (specifically, $\delta \leq 0.9$), we observe that our adaptive-confidence HK model can reach a consensus with $W(T_f) < 1$.
For sufficiently large values of $\gamma$ (specifically, $\gamma \geq 0.05$), even though the nodes in some dyads are not receptive to each other, most nodes (at least 99\% of them, based on our definition of major cluster) still converge to the same final opinion and hence reach a consensus.

For fixed values of $\gamma$ and $c_0$, we observe that $W(T_f)$ tends to decrease as we decrease $\delta$. 
For $\gamma \in \{0.05, 0.1\}$, our simulations always reach a consensus state. In these simulations, for each fixed $\delta$, we observe
that $W(T_f)$ appears to increase monotonically with respect to $c_0$. 
Additionally, for these values of $\gamma$, we observe a transition in $W(T_f)$ as a function of $\delta$. 
For $\delta \leq 0.9$, we observe that $W(T_f) <  1$ and that there is a seemingly linear relationship between $W(T_f)$ and $c_0$. 
When $\delta = 1$, we observe that $W(T_f) = 1$. For $\delta \in \{0.95, 0.99\}$, the behavior of $W(T_f)$ transitions from the $\delta \leq 0.9$ behavior to the $\delta = 1$ behavior as we increase $c_0$. 
This transition between behaviors occurs for smaller $c_0$ for $\delta = 0.99$ than for $\delta = 0.95$.

\begin{figure} [ht]
\label{fig:HK_complete_time}
\includegraphics[width=\columnwidth]{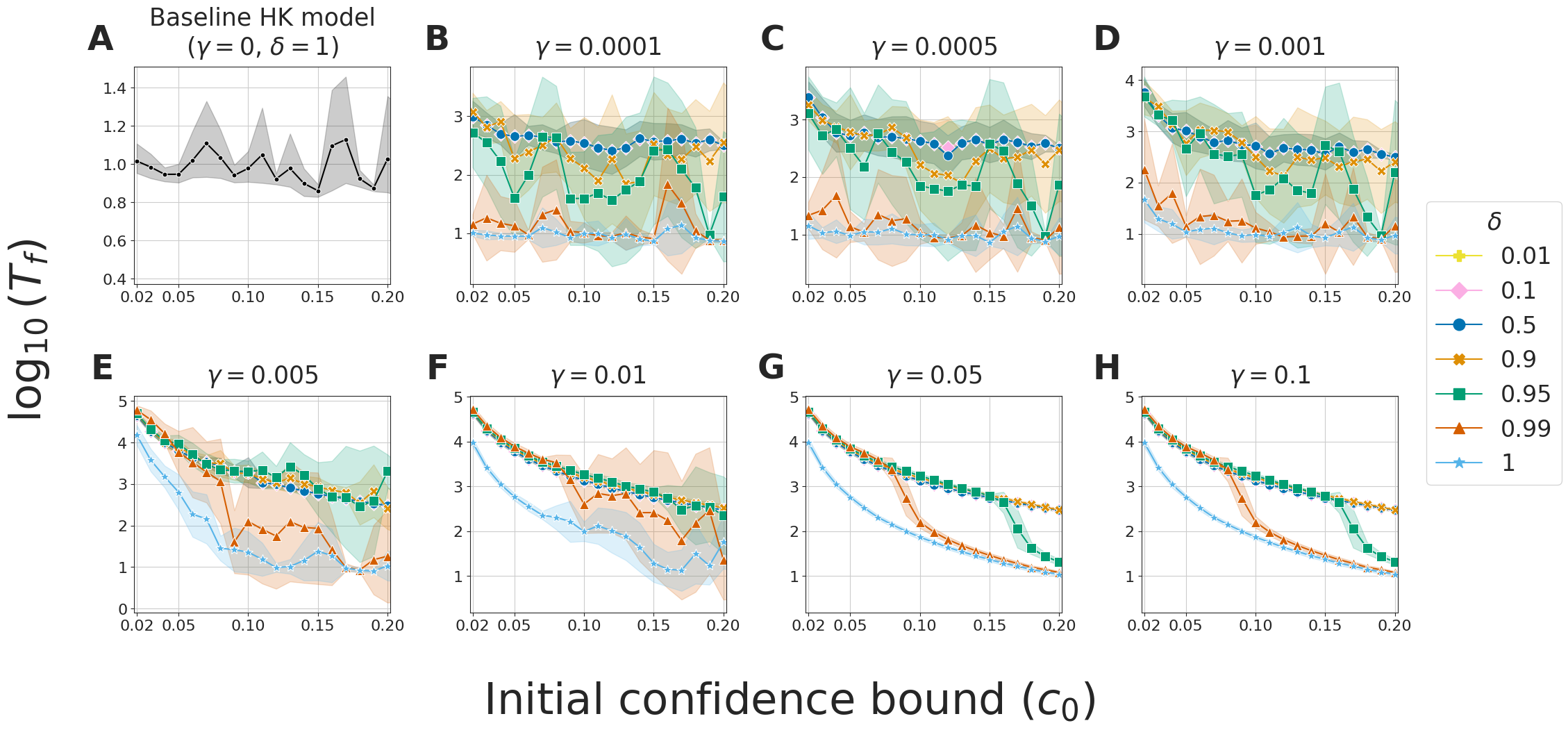}
\caption{The convergence times (in terms of the number of time steps) on a logarithmic scale in simulations of our adaptive-confidence HK model on a 1000-node complete graph for various combinations of the BCM parameters $\gamma$, $\delta$, and $c_0$.}
\end{figure}

In \cref{fig:HK_complete_time}, for fixed $c_0$, we observe that our adaptive-confidence HK model takes longer to converge than the baseline HK model.
For a 1000-node complete graph and fixed BCM parameters (i.e., $\gamma$, $\delta$, and $c_0$), we observe that the logarithm $\log_{10}(T_f)$ of the convergence time $T_f$
for our adaptive-confidence HK model can be up to 4 more than the logarithm of the convergence time for the baseline HK model.
That is, the convergence time can be as much as $10^4$ times larger.
The convergence time tends to increase as either (1) we increase $\gamma$ for fixed $\delta$ and $c_0$ or (2) we decrease $\delta$ for fixed $\gamma$ and $c_0$.
For large values of $\gamma$ (as is especially evident for $\gamma \in \{0.05,0.1\}$), the convergence time decreases with $c_0$. As with $W(T_f)$, for these values of $\gamma$, we observe a transition in the convergence time as a function of $\delta$. 
In \cref{fig:HK_complete_time}, the curves of $\log_{10}(T_f)$ versus $c_0$ for $\delta \leq 0.9$ overlay each other and indicate larger convergence times than the curve for $\delta = 1$. The curves for $\delta = 0.95$ and $\delta = 0.99$ transition from the $\delta \leq 0.9$ behavior to the $\delta = 1$ behavior as we increase $c_0$. 
By contrast, for the baseline HK model and for our model with small values of $\gamma$, we observe no clear pattern between the convergence time and initial confidence bound.
When our adaptive-confidence HK model reaches a consensus state, we observe from the behavior of $W(T_f)$ (see \cref{fig:HK_complete_WT}) and the convergence times (see \cref{fig:HK_complete_time}) in our simulations that there is qualitative transition in the model behavior as we vary $\delta$. 
We are not aware of previous discussions of similar transitions in variants of the HK model.

% ----------------------------------------
\subsubsection{Erd\H{o}s--R\'enyi (ER) and two-community stochastic-block-model (SBM) graphs}\label{sec:HK_ER_SBM} 

We now discuss our simulations of our adaptive-confidence HK model on $G(N,p)$ ER random graphs and two-community SBM random graphs. 
As in our simulations on the complete graph, we observe the trends in \cref{tab:HK_trends}.
We briefly discuss how our observations for these random graphs differ from our observations for the complete graph.
We give additional details about our ER simulations in \cref{sec:HK_ER_appx}, and we give additional details about our SBM simulations in \cref{sec:HK_SBM_appx}.

For fixed BCM parameters (namely, $\gamma$, $\delta$, and $c_0$),
we tend to observe fewer major opinion clusters for $G(1000, 0.1)$ graphs than for the 1000-node complete graph. 
Additionally, for small initial confidence bounds (specifically, $c_0 \leq 0.04$), we observe more minor clusters for the $G(1000, 0.1)$ graphs than the $G(1000, 0.5)$ graphs and the 1000-node complete graph. (For $G(1000, 0.1)$ graphs, once we take the mean for each BCM parameter set, we sometimes observe as many as $20$ minor clusters.) The expected mean degree of a $G(N,p)$ ER graph is $p(N - 1)$ \cite{newman2018}. Therefore, for small probability $p$, we expect more nodes to have small degrees. 
For small initial confidence bounds, we hypothesize that many nodes with small degrees 
quickly disconnect to form minor opinion clusters in the effective graph.

For our two-community SBM graphs, for fixed BCM parameters, the numbers of major clusters and Shannon entropies are similar to those for the complete graph. 
Each of our SBM graphs consists of two complete graphs that are joined by a small number of edges (see \cref{sec:simulations_networks}) to yield a two-community structure. It seems that this two-community structure does not significantly impact the simulation results of our adaptive-confidence HK model.

% ----------------------------------------
\subsubsection{\textsc{Facebook100} university networks}\label{sec:HK_facebook}
We now discuss our simulations of our adaptive-confidence HK model on {\sc Facebook100} networks (see \cref{sec:simulations_networks}) \cite{red2011, Facebook100}.
We show plots of the number of major clusters and Shannon entropy for the UC Santa Barbara network. In \cref{appendix:reed_clusters}, we show a plot of the number of major clusters for Reed College. In our \href{https://gitlab.com/graceli1/Adaptive-Confidence-BCM}{code repository}, we include all plots for our simulations on {\sc Facebook100} networks (including plots for the other examined quantities and the other four universities in \cref{tab:network_sizes}).

The six {\sc Facebook100} networks (see \cref{tab:network_sizes}) mostly exhibit the same trends.\footnote{The Reed College network is a notable exception.
For small initial confidence bounds $c_0 \leq 0.04$ and fixed values of the BCM parameters, it tends to have more major clusters and larger Shannon entropies than the other five {\sc Facebook100} networks.
We hypothesize that this observation, which we discuss further in \cref{appendix:reed_clusters}, arises from finite-size effects.}
Except for the trends in Shannon entropy, we observe the same trends (see \cref{tab:HK_trends}) for the {\sc Facebook100} networks that we observed for the synthetic networks. 
For the {\sc Facebook100} networks, most of the final opinion clusters for both our adaptive-confidence HK model and the baseline HK model are minor opinion clusters. 
Our simulations on the UC Santa Barbara network yield more minor clusters than our simulations on the other {\sc Facebook100} networks; when $c_0 = 0.02$ and $\delta \leq 0.9$, the UC Santa Barbara network has more than $4000$ minor clusters.
Our calculation of Shannon entropy (see \cref{eq:entropy}) includes contributions from minor opinion clusters. Therefore, because of the large numbers of minor clusters for the {\sc Facebook100} networks, the Shannon entropy and numbers of major opinion clusters follow different trends. For these networks, they thus give complementary views of opinion fragmentation.

\begin{figure} [ht]
\centering
\includegraphics[width=0.8\columnwidth]{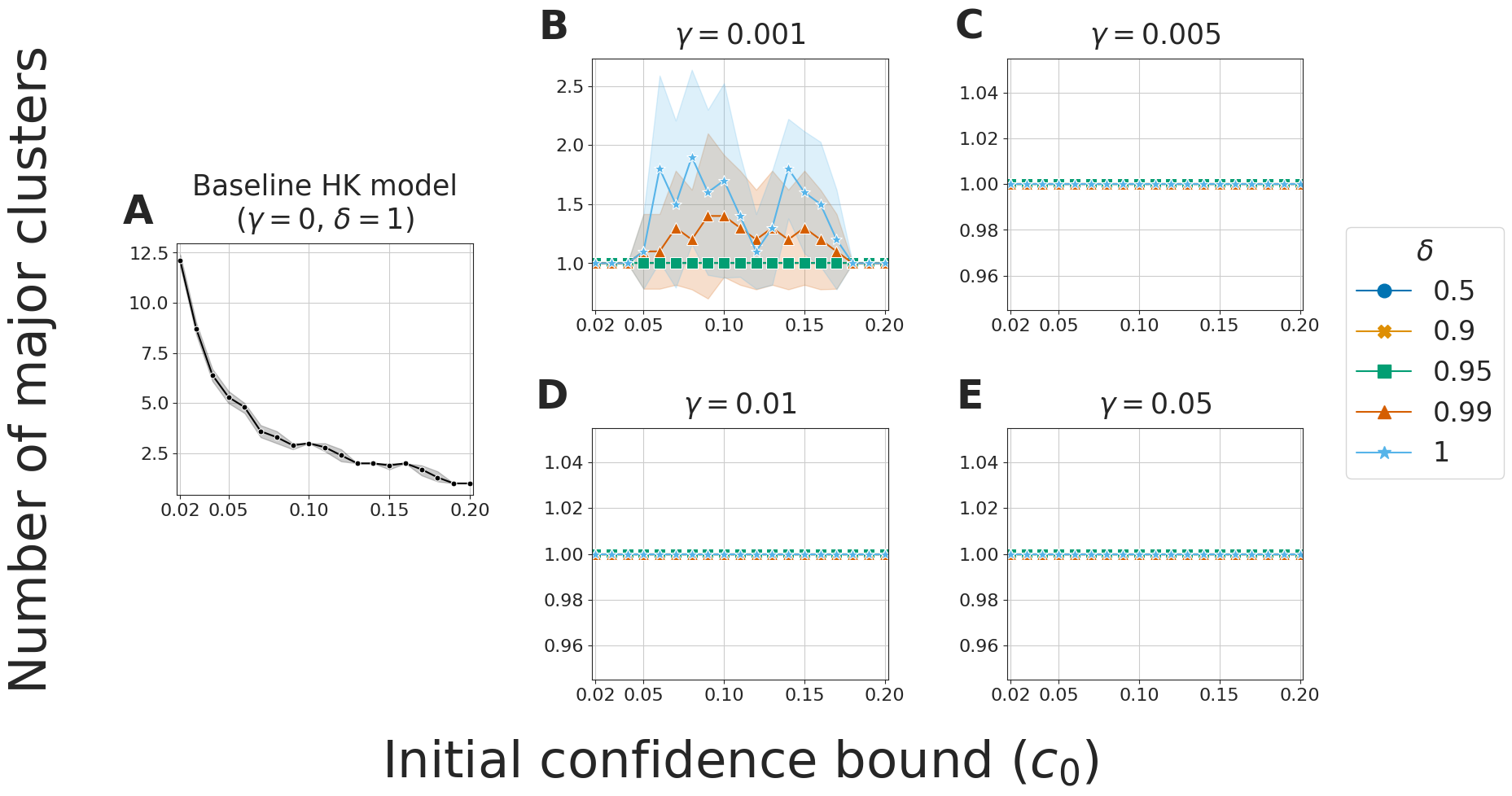}
\caption{The numbers of major clusters in simulations of our adaptive-confidence HK model on the UC Santa Barbara network for various combinations of the BCM parameters $\gamma$, $\delta$, and $c_0$.}
\label{fig:HK_ucsb_clusters}
\end{figure}

\begin{figure} [ht]
\centering
\includegraphics[width=0.8\columnwidth]{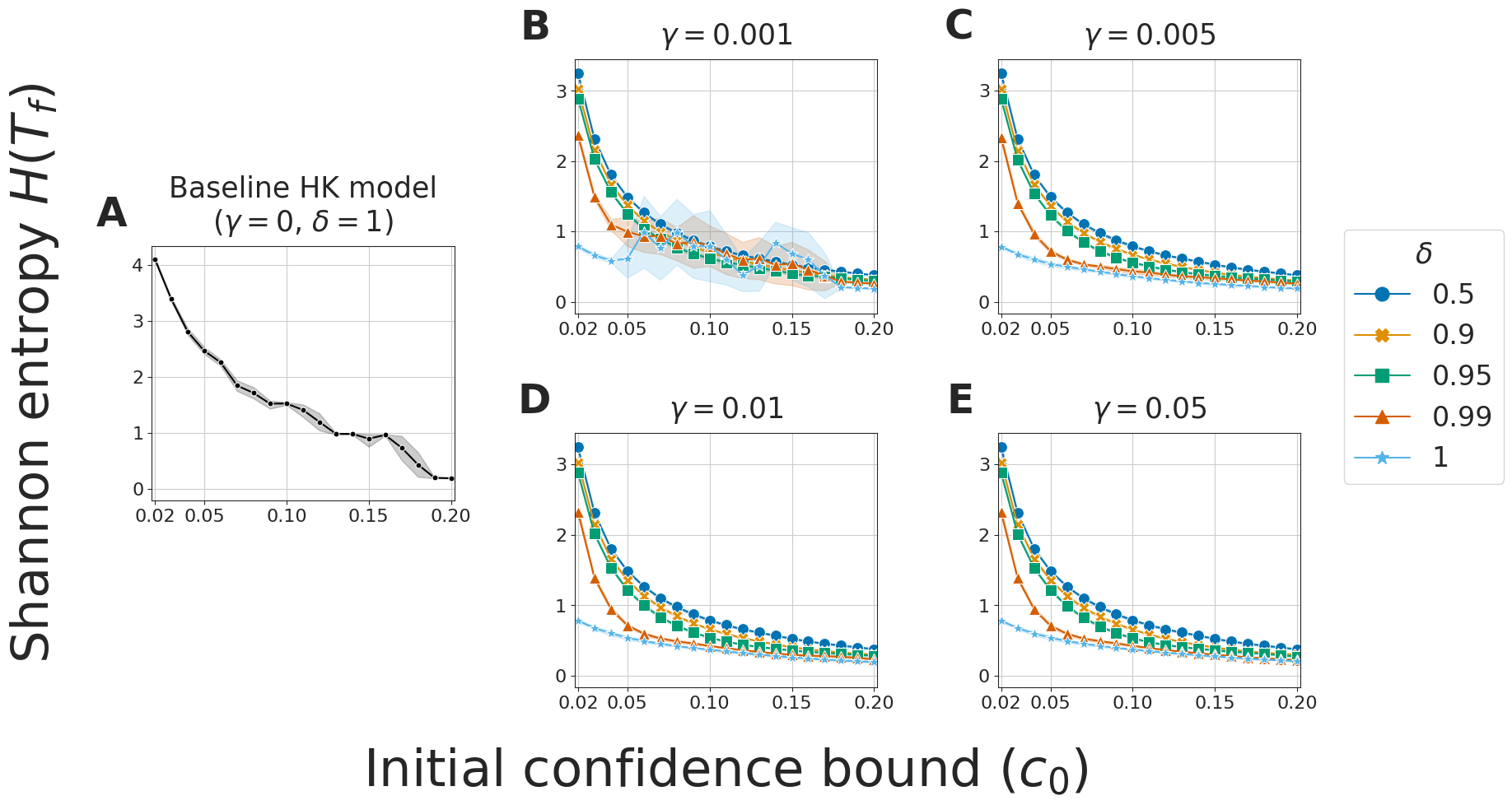}
\caption{The Shannon entropies in simulations of our adaptive-confidence HK model on the UC Santa Barbara network for various combinations of the BCM parameters $\gamma$, $\delta$, and $c_0$.}
\label{fig:HK_ucsb_entropy}
\end{figure}

In \cref{fig:HK_ucsb_clusters}, we observe for the UC Santa Barbara network that our adaptive-confidence HK model always yields consensus (the number of major clusters is exactly $1$) when $\gamma \geq 0.005$.
For these values of $\gamma$, the Shannon entropy (see \cref{fig:HK_ucsb_entropy}) tends to decrease as we increase $c_0$ for fixed values of $\gamma$ and $\delta$.
This trend occurs because the number of minor opinion clusters also tends to decrease as we increase $c_0$ for fixed values of $\gamma$ and $\delta$.
This observation contrasts with our simulations of our adaptive-confidence HK model on synthetic networks (see \cref{sec:HK_complete,sec:HK_ER_SBM}), for which we observed that the Shannon entropy follows similar trends as the number of major clusters as we vary one of $\gamma$, $\delta$, or $c_0$ while fixing the other BCM parameters.
We believe that one reason for this difference is that the {\sc Facebook100} networks have many small-degree nodes, which allow more minor opinion clusters to form.

% ----------------------------------------
% ----------------------------------------
\subsection{Adaptive-confidence DW model}
We now briefly discuss our simulations of our adaptive-confidence DW model on a 100-node complete graph and the {\sc NetScience} network \cite{netscience}. In \cref{tab:DW_trends}, we summarize the trends that we observe in these simulations.
Because of the long computation times, we consider much smaller graphs and fewer BCM parameter values for our adaptive-confidence DW model than we did for our adaptive-confidence HK model.
Notably, the value of the compromise parameter $\mu$ (which is in the DW models but is not the HK models) affects our simulation results.

Our adaptive-confidence DW model yields some of the same trends that we obtained in our adaptive-confidence HK model. 
One of these trends is that, for both the 100-node complete graph and the {\sc NetScience} network,  
the weighted-average edge fraction $W(T_f) < 1$ for some BCM parameters.
Another trend is that our adaptive-confidence DW model has longer convergence times than the baseline DW model in our simulations on the 100-node complete graph.
Additionally, for the 100-node complete graph, we observe less opinion fragmentation as we increase $\gamma$.
We do not observe these latter two trends in our simulations on the {\sc NetScience} network.
In \cref{appendix:DW_simulations}, we give a detailed discussion of the results of our simulations of our adaptive-confidence DW model.

\begin{table}[ht]
\centering
\caption{\label{tab:DW_trends} Summary of the observed trends in our adaptive-confidence DW model.
}
\noindent
\begin{tabularx}{\textwidth}{qt}
\hline\hline
Quantity & Trends \\\hline
Convergence Time
&
$\bullet$ For the complete graph, 
for fixed values of the compromise parameter $\mu$ and initial confidence bound $c_0 \leq 0.3$, our adaptive-confidence DW model tends to converge more slowly than the baseline DW model. \smallskip \newline
$\bullet$ For the {\sc NetScience} network, for fixed values of $\mu$ and $c_0$,
our adaptive-confidence DW model and the baseline DW model have similar convergence times.
\\ \hline
Number of \newline Major Clusters 
& 
$\bullet$ For the complete graph, when we fix the other BCM parameters, 
we (1) tend to observe fewer major clusters as we increase {the confidence-increase parameter} $\gamma$ and 
(2) observe little effect on the numbers of major clusters 
when we vary the confidence-decrease parameter $\delta$.
\smallskip \newline
$\bullet$ For the complete graph, for a fixed value of $c_0 \leq 0.3$, 
our adaptive-confidence DW model yields fewer major clusters when $\mu = 0.1$ than when $\mu \in \{0.3, 0.5\}$. The baseline DW model does not have this behavior. \smallskip \newline
$\bullet$ For the {\sc NetScience} network, for a fixed value of $c_0$, our adaptive-confidence DW model yields at least as many major clusters as the baseline DW model. For this network, $\mu$ has little effect on the number of major clusters. 
\\ \hline
$W(T_f)$ 
&
$\bullet$ The baseline DW model always yields $W(T_f) = 1$. 
Our adaptive-confidence DW model also yields $W(T_f) = 1$ for the complete graph with $c_0 \geq 0.4$ and the {\sc NetScience} network with $c_0 \in \{0.8, 0.9\}$.
\smallskip \newline
$\bullet$  
When $W(T_f) < 1$, for fixed values of the parameters $\gamma$, $\delta$, and $c_0$,
decreasing $\mu$ tends to also decrease $W(T_f)$ for both the
complete graph and the {\sc NetScience} network.
\\ \hline\hline
\end{tabularx}
\end{table}

%%%%%%

%---------
\section{Conclusions and discussion}\label{sec:discussion}
\subsection{Summary and discussion of our results}\label{sec:discussion_summary}
We developed two bounded-confidence models (BCMs) --- a synchronously-updating one that generalizes the Hegselmann--Krause (HK) model and an asynchronously-updating one that generalizes the Deffuant--Weisbuch (DW) model --- with adaptive confidence bounds. The confidence bounds in our adaptive-confidence BCMs are distinct for each dyad (i.e., pair of adjacent nodes) of a network and change when nodes interact with each other. 
One can interpret the changes in confidence bounds as changes in receptiveness between nodes.
We demonstrated that incorporating time-dependent, adaptive confidence bounds in {our} BCMs yields a variety of interesting behaviors, such as adjacent nodes that converge to the same limit opinion but are eventually unreceptive to each other.

For both our adaptive-confidence HK model and our adaptive-confidence DW model, we proved convergence properties for the dyadic confidence bounds and the limiting behaviors of effective graphs, which track which nodes of a network are able to influence each other.
We demonstrated using numerical simulations that our BCMs have fewer major opinion clusters and take longer to converge than the associated baseline BCMs.
See \cref{tab:HK_trends} for a summary of the trends in our adaptive-confidence HK model, and see \cref{tab:DW_trends} for a summary of the trends in our adaptive-confidence DW model.

The results of our numerical simulations of our adaptive-confidence BCMs complement our theoretical results (which informed the stopping criteria in our computations).
For our adaptive-confidence HK model, we proved that (1) all dyadic confidence bounds must converge either to $0$ or to $1$ and (2)
the dyadic confidence bounds between nodes in different limit opinion clusters must converge to $0$ (see \cref{thm:convergence_HK}).
For our adaptive-confidence DW model, we proved that analogous results hold
almost surely (see \cref{thm:convergence_DW}).
However, in both of our adaptive-confidence BCMs, the dyadic confidence bounds between nodes in the same limit opinion cluster do not necessarily converge to $1$, 
as it is possible for them to instead converge to $0$.
Indeed, when the confidence-decrease parameter $\delta < 1$, our numerical simulations of both of our adaptive-confidence BCMs demonstrate for a wide range of the other BCM parameter values that some dyads in the same final opinion cluster have confidence bounds that converge to $0$.
Although the nodes in these dyads are unreceptive to each other, they still converge to the same opinion.
The nodes in these dyads do not have an edge between them in the final effective graph, so the final opinion clusters (i.e., the connected components of the final effective graph) in our BCMs can have a richer structure than those in the baseline BCMs.

%-----------------------------
\subsection{Future work}
Our investigation lays groundwork and provides a point of comparison for the study of more complicated adaptive-confidence mechanisms in BCMs. 
Future investigations of adaptive-confidence BCMs include establishing additional theoretical guarantees, examining and validating such BCMs in sociological contexts, and generalizing these models in various ways.

It is worthwhile to further explore the theoretical guarantees of our adaptive-confidence BCMs. 
We showed (see \cref{thm:effgraph_DW}) that, almost surely,
the effective graph in our adaptive-confidence DW model eventually only has edges between nodes in the same limit opinion cluster. 
However, unlike for our adaptive-confidence HK model (see \cref{thm:effgraph_HK}), we did not prove any guarantee that the effective graph of our adaptive-confidence DW model is eventually constant (not even almost surely). Further theoretical analysis of our adaptive-confidence DW model can help strengthen knowledge of its properties, including the structural properties of the limit effective graphs.

It is also relevant to analytically and numerically study the mutual receptiveness of nodes in our BCMs when they reach a consensus state.
In our numerical simulations of our adaptive-confidence BCMs, when the confidence-decrease parameter $\delta < 1$, some adjacent nodes in the same final opinion cluster eventually are not receptive to each other.
More specifically, our numerical simulations suggest that some
adjacent nodes can converge to the same limit opinion without having an edge between them in the limit effective graph.
One can explore this behavior of our BCMs and determine how the model parameters influence the existence of edges between adjacent nodes with the same limit opinion in limit effective graphs.
It is also interesting to examine the effects of different initial conditions. In our simulations, we observed for uniformly distributed initial opinions that our adaptive-confidence BCMs yield fewer major opinion clusters than the corresponding baseline models. It is worth exploring whether or not this is also the case when different network communities have different initial opinion distributions.

It is also important to consider how the behaviors of our BCMs connect to real-life social situations.  
One can interpret the opinion values in our models as representing outwardly expressed opinions, which may differ from internally held beliefs \cite{kuran1995}. The achievement of a ``consensus'' can represent agents arriving at the same outwardly expressed behavior or decision, rather than achieving an actual agreement of their internal values \cite{horowitz1962}.
Researchers have studied models with both internal and expressed opinions \cite{noorazar_recent_2020, cheng2020, hou2021}, and one can incorporate such considerations into adaptive-confidence BCMs.

In our adaptive-confidence BCMs, adjacent agents that are unreceptive to each other's opinions can still interact with each other. 
Alternatively, a pair of agents can eventually stop interacting with each other --- effectively changing the network structure --- after repeated negative interactions.
Researchers have modeled such ideas, along with network restructuring to consider new social interactions,
using adaptive networks with edge rewiring \cite{kan_adaptive_2022, pansanella2022}.
A possible area of further study is the investigation of which models effectively have ``mediator'' nodes that assist in bringing 
together the opinions of agents that are unreceptive to each other or no longer interact. 
If there are such mediator nodes, one can examine whether or not they share common characteristics or are identifiable from network structure and initial agent opinions.

There are many possible areas to explore in the study of adaptive opinion models. In research on opinion dynamics, it is important to incorporate network adaptivity, which provides fertile ground for theoretical, computational, and empirical investigations of opinion dynamics.

%-----------------------------
%-----------------------------
\appendix

%-----------------------------
%-----------------------------
\section{Proofs of our theoretical results for our adaptive-confidence DW model}\label{appendix:proofs}
We now prove the results for our adaptive-confidence DW model that we presented in \cref{sec:adaptive-confidenceDW_theorems}.

%%%%%

%-----------------------------
\subsection{Proofs of our confidence-bound results}\label{appendix:proof_cij_DW}
We first prove \cref{lemma:monotone_DW}, which states that each confidence bound $c_{ij}(t)$ is eventually monotone.

\begin{proof}[Proof of \cref{lemma:monotone_DW}]

We first consider $c_{ij}(t)$ for adjacent nodes, $i$ and $j$, that are in different limit opinion clusters (i.e., $x^i \neq x^j$). 
Choose a time $T$ such that the inequalities
\begin{align}
	|x_k(t) - x^k| &< \frac{1}{4}\min_{x^m\neq x^{n}}|x^m - x^{n}|\,, \label{eq:DW_monotone_condition1} \\
	|x_k(t) - x_k(t')| &< \frac{\mu}{4} \min_{x^m\neq x^{n}}|x^m - x^{n}|  \label{eq:DW_monotone_condition2}
\end{align}
hold for each node $k$ and all times $t' > t \geq T$.

We claim that $c_{ij}(t)$ is monotone decreasing (i.e., $c_{ij}(t+1) \leq c_{ij}(t)$) for all $t \geq T$. Note that 
\begin{equation}\label{eq:DW_obvious_difference}
	|x^i - x^j| \geq \min_{x^m\neq x^{n}}|x^m - x^{n}|\,.  
\end{equation}
By the triangle inequality and \cref{eq:DW_monotone_condition1}, we have 
\begin{align*}
	|x^i - x^j| &\leq |x^i - x_i(t)| + |x_i(t) - x_j(t)| + |x_j(t) - x^j|\\
	&< \frac{1}{2}\min_{x^m\neq x^{n}}(|x^m - x^{n}|) + |x_i(t) - x_j(t)|\,.  
\end{align*}
Rearranging terms and using \cref{eq:DW_obvious_difference} yields
\begin{equation}\label{eq:DW_contradiction1}
	|x_i(t) - x_j(t)| > \frac{1}{2}\min_{x^m\neq x^{n}}|x^m - x ^{n}|\,.
\end{equation}

Suppose that $c_{ij}(t)$ increases (i.e., $c_{ij}(t+1) > c_{ij}(t)$) at time $t \geq T$. 
This implies that $x_j(t + 1) = x_j(t) + \mu(x_i(t) - x_j(t))$,
which in turn implies that 
\begin{equation}\label{eq:update_diff_DW}
	|x_j(t+1)-x_j(t)|= \mu|x_i(t) - x_j(t)| \,.
\end{equation}
By \cref{eq:DW_monotone_condition2}, we have 
\begin{equation}\label{eq:DW_monotone_condition2_consequence}
	|x_j(t+1) - x_j(t)| < \frac{\mu}{4}\min_{x^m\neq x^{n}}|x^m - x^{n}|\,,
\end{equation}
and \cref{eq:update_diff_DW,eq:DW_monotone_condition2_consequence} together imply that
\begin{equation}\label{eq:DW_contradiction2}
    |x_i(t) - x_j(t)| < \frac{1}{4}\min_{x^m\neq x^{n}}|x^m - x^{n}|\,.
\end{equation}
The inequalities \cref{eq:DW_contradiction1} and \cref{eq:DW_contradiction2} cannot hold simultaneously.
Therefore, any interactions between nodes $i$ and $j$ for times $t \geq T$ must result in a decrease of $c_{ij}$. 
Consequently, for all adjacent nodes $i$ and $j$ from distinct limit opinion clusters, $c_{ij}$ is monotone decreasing (i.e., $c_{ij}(t + 1) \leq c_{ij}(t)$) for all $t \geq T$.

Now consider adjacent nodes, $i$ and $j$, that are in the same limit opinion cluster (i.e., $x^i = x^j$).
Choose a time $T > 0$ such that 
\begin{equation}\label{eq:DW_monotone_condition3}
	|x_k(t) - x^k| < \frac{\gamma}{2}
\end{equation}
for each node $k$ and all times $t \geq T$.
We claim that there is some time $T_{ij} \geq T$ such that either $c_{ij}(t)$ is monotone decreasing (i.e., $c_{ij}(t + 1) \leq c_{ij}(t)$) or it is monotone increasing (i.e., $c_{ij}(t + 1) \geq c_{ij}(t)$) for all $t \geq T_{ij}$.

If $c_{ij}(t)$ is monotone decreasing for all $t \geq T$, choose $T_{ij} = T$. 
If $c_{ij}(t)$ is not monotone decreasing for all $t \geq T$, there must exist a time $T_{ij} \geq T$ at which $|x_i(T_{ij}) - x_j(T_{ij})| < c_{ij}(T_{ij})$. This implies that
\begin{equation}
	c_{ij}(T_{ij} + 1) = c_{ij}(T_{ij}) + \gamma(1 - c_{ij}(T_{ij})) \geq \gamma\,.
\end{equation}
We claim that $c_{ij}(t)$ only increases or remains constant for times $t \geq T_{ij}$.
By \cref{eq:DW_monotone_condition3}, we have
\begin{equation}
	|x_k(t) - x_{k'}(t)| 
    \leq |x_k(t) - x^k| + |x^k - x^{k'}| + |x^{k'} - x_{k'}(t)|
    <\gamma
\end{equation}
for each node pair $k$ and $k'$ with $x^k = x^{k'}$ and all times $t \geq T$.
Therefore, 
\begin{equation}
	|x_i(t) - x_j(t)| < \gamma 
\end{equation}
for all times $t \geq T_{ij} \geq T$,
which implies that subsequent interactions between nodes $i$ and $j$ 
increase $c_{ij}(t)$ (because $c_{ij}(t) \geq \gamma$). 
Consequently,
if $c_{ij}(t)$ increases at a certain time $T_{ij} \geq T$, then it subsequently either increases or remains constant. 
If $c_{ij}(t)$ never increases after time $T$, then by definition it is eventually monotone decreasing. This implies that $c_{ij}(t)$ is either eventually monotone increasing (i.e., $c_{ij}(t_2) \geq c_{ij}(t_1)$ for all $t_2 > t_1 \geq T$) or eventually monotone decreasing. 

\vspace{-22pt}
\[\] %The proof qed box is showing up on the last equation environment before the proof ends. We include a blank equation here to get the symbol to show up at the end.

\end{proof}

\medskip

We now prove \cref{lemma:convergence_DW}, which states that if $c_{ij}(t)$ converges, then its limit $c^{ij} = \lim\limits_{t\to\infty} c_{ij}(t)$ is either $0$ or $1$ almost surely.

\begin{proof}[Proof of \cref{lemma:convergence_DW}]

We prove this lemma using an argument that is similar to the one that we used to prove \cref{lemma:convergence_HK}. 
Unlike in an HK model, 
adjacent nodes in a DW model need not interact with each other at each discrete time.
In fact, it is possible (although it occurs with probability $0$) that 
there exists a pair of adjacent nodes that only interact a finite number of times.

Given $\epsilon > 0$, choose a time $T$ such that 
the inequalities
\begin{align}
    |c_{ij}(t) - c^{ij}| &<  \epsilon/2\,, 	\label{eq:DW_cij_condition1} \\
    |c_{ij}(t_1) - c_{ij}(t_2)| &< \frac{1}{2}\left(\min\{1 - \delta, \gamma\}\right) \epsilon	\label{eq:DW_cij_condition2}
\end{align}
hold for all times $t,t_1,t_2\geq T$.
Suppose that we choose the adjacent nodes $i$ and $j$ to interact at some time $t \geq T$. 
It then follows that
either $c_{ij}(t + 1) = \delta c_{ij}(t)$ or $c_{ij}(t + 1) = c_{ij}(t)+\gamma(1-c_{ij}(t))$.

Suppose that $c_{ij}(t + 1) = \delta c_{ij}(t)$. In this case, we claim that $c^{ij} = 0$. To verify this claim, note that $c_{ij}(t) - c_{ij}(t + 1) = (1 - \delta)c_{ij}(t)$. 
We know from \cref{eq:DW_cij_condition2} that $c_{ij}(t) - c_{ij}(t + 1) < \frac{1}{2}(1 - \delta)\epsilon$, so
$c_{ij}(t) < \epsilon/2$. Therefore, with \cref{eq:DW_cij_condition1}, we obtain
\begin{align*}
	0 \leq c^{ij} &\leq |c^{ij} - c_{ij}(t)| + |c_{ij}(t)| \\
	&< \epsilon/2 + \epsilon/2\\
	&= \epsilon\,, 
\end{align*}
which implies that $c^{ij} = 0$.

Now suppose that $c_{ij}(t + 1) = c_{ij}(t) + \gamma(1-c_{ij}(t))$. Rearranging terms yields 
$c_{ij}(t + 1) - c_{ij}(t) = \gamma(1 - c_{ij}(t)) < \frac{1}{2}\gamma\epsilon$, which implies that $1 - c_{ij}(t) < \epsilon/2$. Therefore, 
\begin{align*}
	0 \leq 1 - c^{ij} &\leq |1 - c_{ij}(t)| + |c_{ij}(t) - c^{ij}| \\ 
	&< \epsilon/2 + \epsilon/2 \\
	&= \epsilon\,,
\end{align*}
which implies that $c^{ij} = 1$\,. 

Consequently, if nodes $i$ and $j$ interact infinitely often, $c^{ij}$ must be either $0$ or $1$. By the Borel--Cantelli lemma, nodes $i$ and $j$ interact infinitely many times with probability $1$. Therefore, it is almost surely the case that either $c^{ij} = 0$ or $c^{ij} = 1$.
\end{proof}

%-----------------------------
\subsection{Proof of the effective-graph theorem for our adaptive-confidence DW model}\label{appendix:proof_eff_DW}
We now prove \cref{thm:effgraph_DW}, which is our main result about effective graphs in our adaptive-confidence DW model. It states that, almost surely, an effective graph in our adaptive-confidence DW model eventually has edges only between adjacent nodes in the same limit opinion cluster. 

\begin{proof}[Proof of \cref{thm:effgraph_DW}]

By \cref{thm:convergence_DW}, for adjacent nodes $i$ and $j$ that are in different limit opinion clusters, $c_{ij}(t)$ almost surely converges to $0$. 
Therefore, almost surely, there is some $T_1$ such that 
\begin{equation}\label{eq:confidence_bound_zero_DW_effective_graph}
	c_{ij}(t) < \frac{1}{2}\min_{x^m \neq x^{n}}|x^m - x^{n}|
\end{equation}
for all times $t \geq T_1$. 
We also choose $T_2$ such that
\begin{equation}\label{eq:nodes_close_to_final_opinions_DW_effective_graph}
	|x_k(t) - x^k|<\frac{1}{4}\min_{x^m \neq x^{n}}|x^m - x^{n}|
\end{equation}
for each node $k$ and all times $t \geq T_2$.

Let $T = \max\{T_1,T_2\}$, and fix adjacent nodes $i$ and $j$ that are in different limit opinion clusters. 
The time $T$ exists almost surely because $T_1$ exists almost surely.
For all times $t \geq T$, the inequality \cref{eq:nodes_close_to_final_opinions_DW_effective_graph} implies that 
\begin{align*}
	|x^i - x^j| &\leq |x_i(t) - x^i| + |x_i(t) - x_j(t)| + |x_j(t) - x^j| \\ 
	&\leq \frac{1}{2}\min_{x^m \neq x^{n}}|x^m - x^{n}| + |x_i(t) - x_j(t)|\,.
\end{align*}
Because $\min\limits_{x^m \neq x^{n}}|x^m - x^{n}| \leq |x^i - x^j|$, it follows that 
\begin{equation}
	\min_{x^m\neq x^{n}}|x^m - x^{n}| \leq \frac{1}{2}\min_{x^m \neq x^{n}}|x^m - x^{n}| + |x_i(t) - x_j(t)|\,.
\end{equation}
Therefore, with \cref{eq:confidence_bound_zero_DW_effective_graph}, we obtain
\begin{align*}
	c_{ij}(t) < \frac{1}{2}\min_{x^k \neq x^{k'}}|x^k - x^{k'}| \leq |x_i(t) - x_j(t)|\,.
\end{align*}
That is, $|x_i(t) - x_j(t)| \geq c_{ij}(t)$ for all $t\geq T$, so
the edge $(i,j)$ is not in the effective graph at time $t$ for all $t \geq T$. Therefore, the only edges in the effective graph for times $t \geq T$ 
are between nodes in the same limit opinion cluster. 
\vspace{-20pt}
\[\] %The proof qed box is showing up on the last equation environment before the proof ends. We include a blank equation here to get the symbol to show up at the end.

\end{proof}

The effective-graph theorem for our adaptive-confidence DW model (see \cref{thm:effgraph_DW}) is weaker than that for our adaptive-confidence HK model (see \cref{thm:effgraph_HK}). In particular, we are unable to conclude for the adaptive-confidence DW model that the effective graph is eventually constant (or even almost surely eventually constant). The obstruction to obtaining such a guarantee 
is the stochasticity that arises from the asynchronous opinion updating of the adaptive-confidence DW model. 
In particular, consider adjacent nodes $i$ and $j$ in the same limit opinion cluster.
By \cref{lemma:monotone_DW}, there exists a time $T$ such that $c_{ij}$ is monotone for all times $t \geq T$. 
Suppose that $c_{ij}$ is monotone decreasing for all $t \geq T$. 
If nodes $i$ and $j$ interact at time $t \geq T$, then $c_{ij}(t)$ decreases and the edge $(i,j)$ is not in the effective graph (i.e., $(i,j) \notin E_\mathrm{eff}(t)$). However, if nodes $i$ and $j$ do not interact at time $t$, it is possible that edge $(i,j) \in E_\mathrm{eff}(t)$.(By contrast, for our adaptive-confidence HK model, which updates synchronously, the existence of a time $T$ such that $c_{ij}$ is strictly decreasing for all times $t \geq T$ implies that the edge $(i,j) \notin E_\mathrm{eff}(t)$ for times $t \geq T$.)
Suppose that $t_1, t_2, \ldots$ are successive times at which nodes $i$ and $j$ interact after time $T$. 
For each $s$, it can be the case that both the inequality $|x_i(t_s) - x_j(t_s)| \geq c_{ij}(t_s)$ holds and there is a time $\tilde{t}_s$ between $t_s$ and $t_{s+1}$ such that $|x_i(\tilde{t}_s) - x_j(\tilde{t}_s)| <  c_{ij}(\tilde{t}_s) = c_{ij}(t_{s})$. 
That is, between each pair of interaction times $t_s$ and $t_{s+1}$, the opinions of nodes $i$ and $j$ can (because of other adjacent nodes) first become close enough so that the difference between their opinions is less than their confidence bound and then subsequently become sufficiently far apart so that the difference between their opinions exceeds their confidence bound. 
In this situation, the effective graph is not eventually constant. 

Although the example in the previous paragraph may seem pathological, it is unclear 
whether and how frequently such situations can occur. 
There also may be other scenarios in which an effective graph is not eventually constant. 
This issue does not arise in the proof of \cref{thm:effgraph_HK} because the nodes in each dyad interact at every time in the adaptive-confidence HK model. 

%-----------------------------
%-----------------------------
\section{Proof of the effective-graph theorem for the baseline DW model}\label{appendix:proof_eff_baseline_DW}
We now prove \cref{thm:effgraph_baseline_DW}, 
which is our convergence result for effective graphs in the baseline DW model. 
To do this, we first prove \cref{lemma:baseDW_eff_edges} and \cref{lemma:baseDW_less_than_c}.

\begin{lemma}\label{lemma:baseDW_eff_edges}
Consider the baseline DW model (with update rule \cref{eq:DW_update_rule}). There is a time $T_1$ and there is almost surely a time $T_2$
such that the following statements hold for all adjacent nodes $i$ and $j$.
\begin{itemize}
    \item[(1)] If $|x^i - x^j| < c$, then $|x_i(t) - x_j(t)| < c$ and the edge $(i,j)$ is in the effective graph for all times $t \geq T_1$.
    \item[(2)] If $|x^i - x^j| > c$, then $|x_i(t) - x_j(t)| > c$ and the edge $(i,j)$ is not in the effective graph 
    at any time $t \geq T_1$.
    \item[(3)] If $|x^i - x^j| = c$, then 
    $|x_i(t) - x_j(t)| \geq c$ 
    and the edge $(i,j)$ is not in the effective graph 
    at any time $t \geq T_2$.
\end{itemize}
\end{lemma} 

\begin{proof}
Consider adjacent nodes $i$ and $j$, and let $\Delta_{ij} = |x^i - x^j|$ denote the difference between their opinions.

We first consider the case in which $\Delta_{ij} \neq c$. Choose a time $T_{ij}$ such that
\begin{equation}
    |x_k(t) - x^k| < \frac{1}{2} |c - \Delta_{ij}|
\end{equation}
for node $k \in \{i,j\}$ and all times $t \geq T_{ij}$.

Suppose that $\Delta_{ij} < c$. For all times $t \geq T_{ij}$, we have
\begin{align*}
	|x_i(t) - x_j(t)| &\leq |x_i(t) - x^i| + |x^i - x^j| + |x_j(t) - x^j| \\
		&< 2\left(\frac{1}{2}\right)(c - \Delta_{ij}) + \Delta_{ij} \\
		&= c \,.
\end{align*}
Therefore, the edge $(i,j)$ is in the effective graph for all $t \geq T_{ij}$.

Now suppose that $\Delta_{ij} > c$. Without loss of generality, let $x^i > x^j$. For all times $t \geq T_{ij}$, we have
\begin{align*}
	x_i(t) - x_j(t) &> \left(x^i -  \frac{1}{2} \left|c - \Delta_{ij}\right| \right) 
 - \left(x^j+ \frac{1}{2} \left|c - \Delta_{ij}\right| \right) \\
	&= (x^i - x^j) - \left|c - \Delta_{ij}\right|\\
	&= \Delta_{ij} - \Delta_{ij} + c \\
	&= c \,.
\end{align*}
Therefore, the edge $(i,j)$ is not in the effective graph at any time $t \geq T_{ij}$.

If there are no adjacent nodes $i$ and $j$ with $|x^i - x^j| \neq c$, then let $T_1 = 0$. Otherwise, let
\begin{equation}
	T_1 = \max\limits_{(i,j) \in E} \{T_{ij} \text{~such that~} |x^i - x^j| \neq c\} \,.
\end{equation}
We have shown that statements (1) and (2) hold for all times $t \geq T_1$.

We now consider the case $\Delta_{ij} = c$. Without loss of generality, let $x^i > x^j$. Choose a time $\tilde{T}_{ij}$ so that 
\begin{equation}\label{eq:baseDW_equal_c}
    |x^k - x_k(t)| < \frac{\mu c}{2(1 + 2\mu)}
\end{equation}
for node $k \in \{i,j\}$ and all times $t \geq \tilde{T}_{ij}$.
We will show that, almost surely, there are a finite number of times $t \geq \tilde{T}_{ij}$ such that $|x_i(t) - x_j(t)| < c$. 
Suppose on the contrary that there is a sequence $t_1, t_2, \ldots$
of times such that $t_k \geq \tilde{T}_{ij}$ and $|x_i(t_k) - x_j(t_k)| < c$ for all {$k$}.
At each time $t$, nodes $i$ and $j$ interact with probability $1 / |E| > 0$, where $|E|$ is the number of edges in the graph. Therefore, with probability $1$, nodes $i$ and $j$ interact at some time $t_k \geq \tilde{T}_{ij}$ with
 $|x_i(t_k) - x_j(t_k)| < c$.
Nodes $i$ and $j$ compromise their opinions at time $t_k$, so the inequality \cref{eq:baseDW_equal_c} implies that
\begin{align}  
    |x_i(t + 1) - x_i(t)| = \mu|x_i(t) - x_j(t)| 
    &\geq \mu \left[ \left(x^i - \frac{\mu c}{2(1 + 2\mu)}\right) 
        - \left(x^j + \frac{\mu c}{2(1 + 2\mu)}\right)  \right] \nonumber \\
    &= \mu \left[c - 2\left(\frac{\mu c}{2(1 + 2\mu)}\right) \right] \nonumber \\
    &= \frac{\mu c [(1 + 2\mu) - \mu]}{1 + 2\mu} \nonumber \\ 
    &= (1 + \mu) \frac{\mu c}{1 + 2\mu} \nonumber \\ 
    &> \frac{\mu c}{1 + 2\mu} \label{eq:c_apart_contradict} \,.
\end{align}
From the inequality \cref{eq:baseDW_equal_c}, we have
\begin{align*}  
    |x_i(t + 1) - x_i(t)| \leq |x_i(t + 1) - x^i| + |x^i - x_i(t)| < 2 \left( \frac{\mu c}{2(1 + 2\mu)} \right) = \frac{\mu c}{1 + 2\mu} \,,
\end{align*}
which cannot hold simultaneously with inequality \cref{eq:c_apart_contradict}.
Therefore, with probability $1$, there are a finite number of 
times $t \geq \tilde{T}_{ij}$ such that $|x_i(t) - x_j(t)| < c$. Consequently, there almost surely
exists some time $T_{ij} \geq \tilde{T}_{ij}$ such that $|x_i(t) - x_j(t)| \geq c$ and the edge $(i,j)$ is not in the effective graph for any $t \geq T_{ij}$.

If there are no adjacent nodes $i$ and $j$ with $|x^i - x^j| \neq c$, then let $T_2 = 0$.
Otherwise, let
\begin{equation}
    T_2 = \max\limits_{(i,j)\in E} 
    \{ T_{ij} \text{~such that~} |x^i - x^j| = c \} \,,
\end{equation}
where $T_2$ exists almost surely because each $T_{ij}$ exists almost surely.
We have shown that statement (3) holds for all times $t \geq T_2$ if $T_2$ exists.

\vspace{-22pt}
\[\] %The proof qed box is showing up on the last equation environment before the proof ends. We include a blank equation here to get the symbol to show up at the end.

\end{proof}

\begin{lemma}\label{lemma:baseDW_less_than_c}
For adjacent nodes $i$ and $j$ with $|x^i - x^j| < c$, we have $x^i = x^j$ almost surely.
\end{lemma}

\begin{proof}
Fix adjacent nodes $i$ and $j$ with $|x^i - x^j| < c$, and let $\Delta_{ij} = |x^i - x^j|$ denote the distance between their opinions. Without loss of generality, let $x^i > x^j$.
We want to show that $\Delta_{ij} = 0$ almost surely. Suppose instead that $\Delta_{ij} > 0$.
Fix $\epsilon$ so that $0 < \epsilon < \min\{\frac{1}{4} (c - \Delta_{ij}), \frac{\Delta_{ij}}{2(1+ 1/\mu)}\}$ and choose $T_{ij}$ so that
\begin{equation}\label{eq:baseDW_less_c_ineq1}
	|x^k - x_k(t)| < \epsilon
\end{equation}
for each node $k$ and all times $t \geq T_{ij}$.

By the Borel--Cantelli lemma, there is almost surely some time $t \geq T_{ij}$ at which nodes $i$ and $j$ interact.
The inequality $\epsilon < \frac{1}{4} (c-\Delta_{ij})$ implies that
\begin{align*}
    |x_i(t) - x_j(t)| 
    	&\leq |x_i(t) - x^i| + |x^i - x^j| + |x^j - x_j(t)| \\
    	&< \frac{1}{4} (c -\Delta_{ij}) + \Delta_{ij} + \frac{1}{4} (c-\Delta_{ij})
    	= \frac{1}{2} (\Delta_{ij} + c) \\
    	&< c \,,
\end{align*}
so nodes $i$ and $j$ are receptive to each other at time $t$. 
Consequently, if they interact at time $t$, they update their opinions and
\begin{align}
    x_j(t + 1) &= x_j(t) + \mu [x_i(t) - x_j(t)] \nonumber \\
   	 &\geq x_j(t) + \mu [x^i - \epsilon - (x^j + \epsilon)] \nonumber \\ 
   	 &= x_j(t) + \mu (\Delta_{ij} - 2\epsilon) \nonumber \\
   	 &\geq (x^j - \epsilon) + \mu (\Delta_{ij} - 2\epsilon) \nonumber \\
	    &> x^j + \epsilon \,, \label{eq:baseDW_less_c_ineq2}
\end{align}
where the last inequality holds because 
$\epsilon < \frac{\Delta_{ij}}{2(1+ 1/\mu)}$, which we rearrange to obtain 
$2\epsilon < \mu(\Delta_{ij} - 2\epsilon)$.
The inequality \cref{eq:baseDW_less_c_ineq1} implies that
\begin{equation}
	|x^j - x_j(t + 1)| < \epsilon \,,
\end{equation}
which cannot hold simultaneously with the inequality \cref{eq:baseDW_less_c_ineq2}.
Therefore, if $0 < x^i - x^j < c$, then nodes $i$ and $j$ cannot interact at times $t \geq T_{ij}$.
However, by the Borel--Cantelli lemma, nodes $i$ and $j$ almost surely interact infinitely often. Consequently, $0 < x^i - x^j < c$ with probability $0$. 
Therefore, we almost surely have $x^i = x^j$.

\vspace{-20pt}
\[\] %The proof qed box is showing up on the last equation environment before the proof ends. We include a blank equation here to get the symbol to show up at the end.

\end{proof}

We now use \cref{lemma:baseDW_eff_edges} and \cref{lemma:baseDW_less_than_c} to prove \cref{thm:effgraph_baseline_DW}.

\begin{proof}[Proof of \cref{thm:effgraph_baseline_DW}]

There is a time $T_1$ such that statements (1) and (2) of \cref{lemma:baseDW_eff_edges} hold, and there is almost surely a time $T_2$ such that statement (3) of \cref{lemma:baseDW_eff_edges} holds. 
Therefore, there is almost surely a time $T = \max \{T_1, T_2\}$ such that all three statements (1)--(3) of \cref{lemma:baseDW_eff_edges} hold for all times $t \geq T$.
Consequently, the edges of the effective graph satisfy $E_\mathrm{eff}(t) = E_\mathrm{eff}(T)$ for all $t \geq T$. 
The effective graph is thus eventually constant with respect to time for all $t \geq T$.

Suppose that the limit effective graph $\lim\limits_{t\to\infty} G_\mathrm{eff}(t)$ exists.
For adjacent nodes $i$ and $j$ with the same limit opinion (i.e., $x^i = x^j$), we know that $|x^i - x^j| = 0 < c$. By statement (1) of \cref{lemma:baseDW_eff_edges}, there thus exists a time $T_1$ such that the edge $(i,j)$ is in the effective graph for all times $t \geq T_1$. Therefore, the edge $(i,j)$ is in the limit effective graph.

Now suppose that the edge $(i,j)$ is in the limit effective graph. We seek to show that $x^i = x^j$ almost surely.
Because the edge $(i,j)$ is in the limit effective graph, there exists a time $\tilde{T}$ such that $(i,j) \in E_\mathrm{eff}(t)$ for all times $t \geq \tilde{T}$.
Consequently, by statement (2) of \cref{lemma:baseDW_eff_edges}, it cannot be the case that $|x^i - x^j| > c$. Therefore, by statement (3) of \cref{lemma:baseDW_eff_edges}, it almost surely cannot be the case that $|x^i - x^j| = c$. Consequently, we almost surely have $|x^i - x^j| < c$.
By \cref{lemma:baseDW_less_than_c}, it is almost surely the case that $x^i = x^j$. 

\end{proof}

%-----------------------------
%-----------------------------
\section{Additional results and discussion of our numerical simulations of our adaptive-confidence HK model}\label{appendix:HK_simulations}

In this appendix, we show additional numerical results for our adaptive-confidence HK model on ER graphs (see \cref{sec:HK_ER_appx}), two-community SBM graphs (see \cref{sec:HK_SBM_appx}), and the Reed College network (see \cref{fig:HK_reed_clusters}).
We again examine the numbers of major and minor clusters, the Shannon entropy $H(T_f)$ (see equation \cref{eq:entropy}), the weighted-average edge fraction $W(T_f)$ (see equation \cref{eq:weighted_avg}), and the convergence time $T_f$. 
We consider the values of the BCM parameters (namely, the confidence-increase parameter $\gamma$, confidence-decrease parameter $\delta$, and initial confidence bound $c_0$) in \cref{tab:parameters}.
Each point in our plots is a mean of our numerical simulations for the associated values of the BCM parameter set ($\gamma$, $\delta$, $c_0$).
All plots, including those that we do not include in this appendix, are in our
\href{https://gitlab.com/graceli1/Adaptive-Confidence-BCM}{code repository}.

%%%

\subsection{ER graphs}\label{sec:HK_ER_appx}

We now discuss additional results of our simulations of our adaptive-confidence HK model on $G(N,p)$ ER random graphs. We generate 5 ER random graphs for each value of $p \in \{0.1,0.5\}$. Each point in our plots is a mean of 50 simulations (from 5 random graphs that each have 10 sets of initial opinions). 
For fixed BCM parameters {(namely,} $\gamma$, $\delta$, and $c_0$),
our results for $G(1000, 0.5)$ graphs are more similar than those for $G(1000, 0.1)$ graphs to our results for the 1000-node complete graph.

\begin{figure} [htb]
\label{fig:HK_ER_clusters}
\includegraphics[width=\columnwidth]{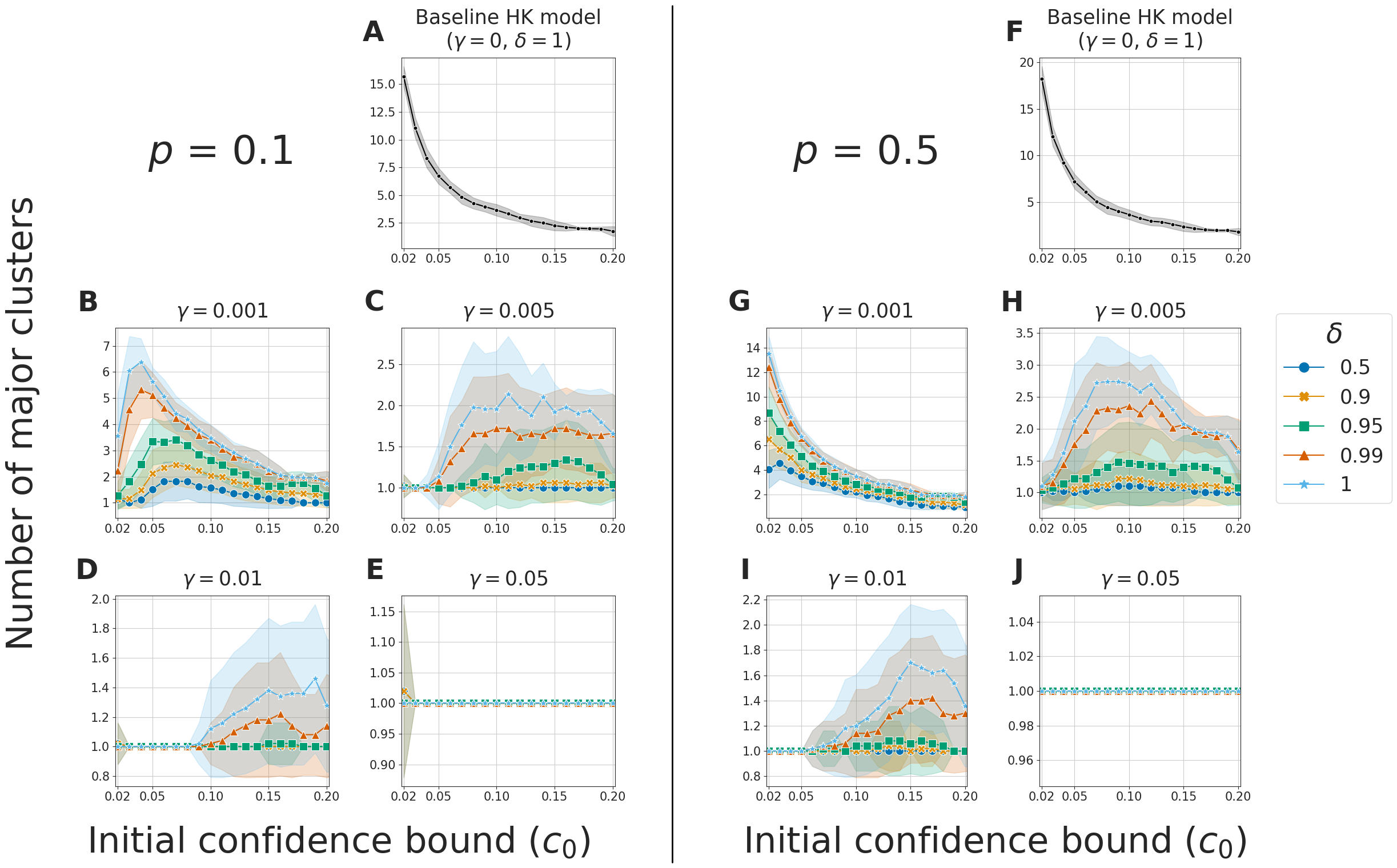}
\caption{The numbers of major clusters in simulations of our adaptive-confidence HK model on $G(1000, p)$ ER random graphs with (A--E) $p = 0.1$ and (F--J) $p = 0.5$ for various combinations of the BCM parameters $\gamma$, $\delta$, and $c_0$.
}
\end{figure}

In \cref{fig:HK_ER_clusters}, we show the numbers of major clusters in our simulations of our adaptive-confidence HK model on ER graphs.
For the 1000-node complete graph and $G(1000, 0.5)$ graphs, for fixed values of $\gamma$ and $\delta$, the number of major clusters tends to decrease as we increase $c_0$.
For the $G(1000, 0.1)$ graphs, when $\gamma = 0.001$, small values of $c_0$ tend to yield few major clusters. As we increase $c_0$, we observe an initial increase in the number of major clusters followed by a decrease in that number. By contrast, for the 1000-node complete graph, small values of $c_0$ tend to yield the most major clusters.
As we discussed in \cref{sec:HK_ER_SBM}, $G(1000, 0.1)$ graphs have more small-degree nodes than the complete graph. These small-degree nodes can easily form minor opinion clusters, especially for small values of $c_0$.
We hypothesize that these minor clusters form quickly in a simulation and that the nodes in them quickly become unreceptive to the other nodes of a network. 
It is thus possible that the nodes that are not in these minor clusters become receptive to fewer neighbors with conflicting opinions.\footnote{When the neighbors to which a node is receptive have very different opinions, recall (see \cref{footnote:conflicting_opinion}) that that node is receptive to ``conflicting'' opinions, resulting in more consensus (i.e., fewer major opinion clusters).}

As we discussed in \cref{sec:HK_ER_SBM}, for small values of $c_0$, our simulations of the adaptive-confidence HK model on the $G(1000, 0.1)$ graphs yield more minor clusters than our simulations on the $G(1000, 0.5)$ graphs and the 1000-node complete graph. Nevertheless, although Shannon entropy (see equation \cref{eq:entropy}) accounts for minor clusters, we still observe that it follows similar trends as the number of major clusters for both $p = 0.1$ and $p = 0.5$. 
Specifically, the Shannon entropy tends to increase as either (1) we decrease $\gamma$ for fixed $\delta$ and $c_0$ or (2) we increase $\delta$ for fixed $\gamma$ and $c_0$.

For our simulations on ER graphs with both $p = 0.1$ and $p = 0.5$, we observe the convergence-time trends in \cref{tab:HK_trends}.
For fixed values of $\gamma$, $\delta$, and $c_0$, the mean convergence time for $p = 0.1$ is at least as long as that for $p = 0.5$. 
Unlike for the complete graph, the ER graphs do not have a clear trend in the dependence of the convergence time either on
$\gamma$ (with fixed $\delta$ and $c_0$) or on
$\delta$ (with fixed $\gamma$ and $c_0$).
As with the 1000-node complete graph,
our fastest convergence times for ER graphs typically occur for $\delta = 1$.
For fixed $\gamma$ and $c_0$, we often observe that the convergence time increases as we decrease $\delta$.
However, we do not always observe this trend; for some values of $\gamma$ and $c_0$, smaller values of $\delta$ yield faster convergence than larger values of $\delta$.

% ----------------------------------------
\subsection{Two-community SBM graphs}\label{sec:HK_SBM_appx}

We now discuss additional results of our simulations of our adaptive-confidence HK model on two-community SBM graphs. Each of our SBM graphs consists of two complete graphs that are joined by a small number of edges (see \cref{sec:simulations_networks}). This yields a two-community structure.

In \cref{fig:HK_SBM_clusters}, we show the numbers of major clusters in our simulations on SBM graphs. 
For fixed values of the BCM parameters (namely, $\gamma$, $\delta$, and $c_0$), these simulations yield similar numbers of major clusters as in our simulations on the 1000-node complete graph (see \cref{fig:HK_complete_clusters}) and $G(1000, 0.5)$ ER graphs (see \cref{fig:HK_ER_clusters}).
We observe few minor clusters; for each BCM parameter set, the mean number of minor clusters is bounded above by $3$.
Consequently, the Shannon entropy and the number of major clusters follow similar trends.

\begin{figure} [htb]
\centering
\includegraphics[width=0.8\columnwidth]{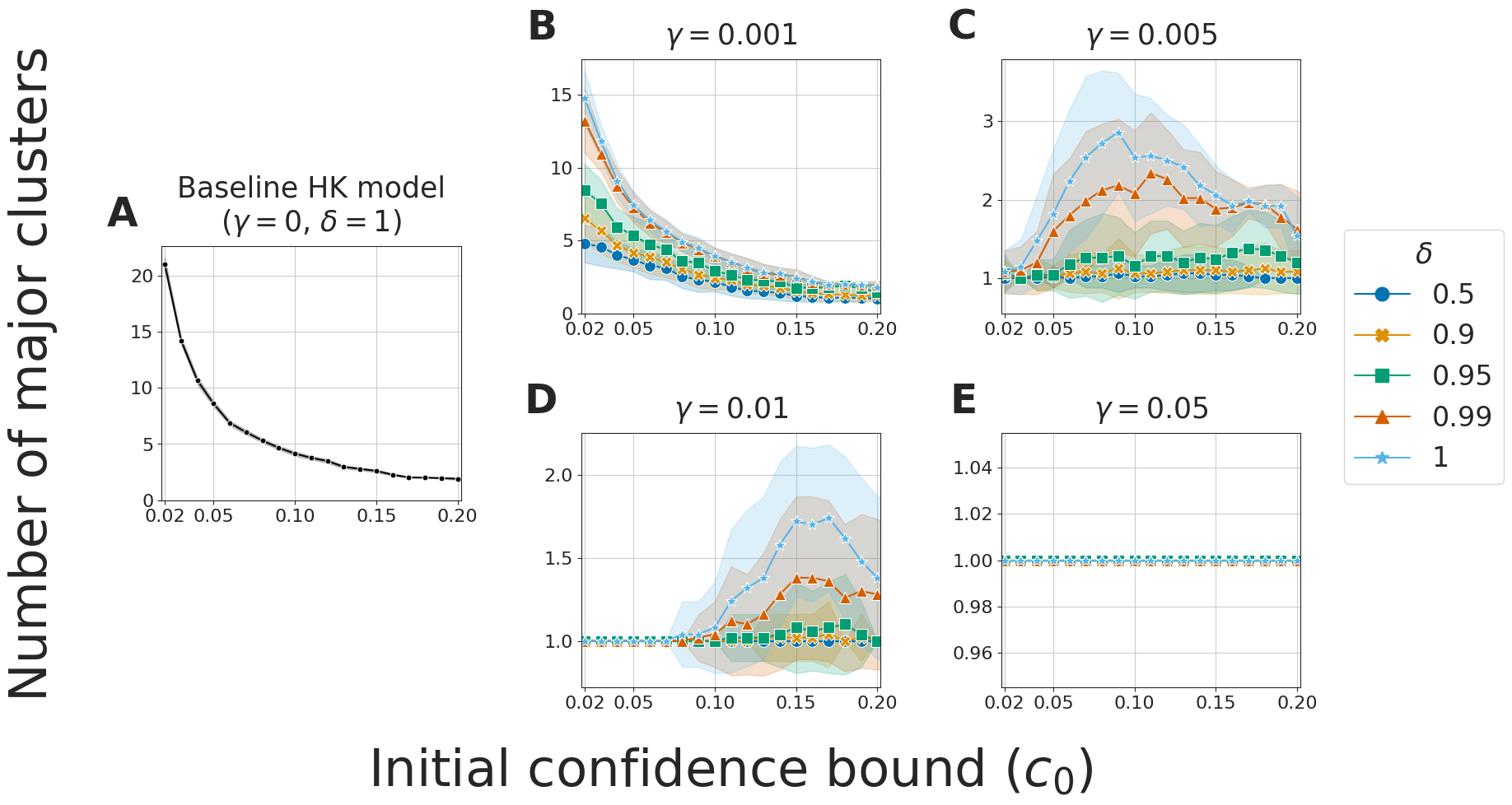}
\caption{The numbers of major clusters in simulations of our adaptive-confidence HK model on 1000-node SBM random graphs with connection probabilities $p_{aa} = p_{bb} = 1$ and $p_{ab} = 0.01$ for various combinations of the BCM parameters $\gamma$, $\delta$, and $c_0$.}
\label{fig:HK_SBM_clusters}
\end{figure}

The convergence times in our simulations on SBM graphs follow the trends in \cref{tab:HK_trends}. 
For fixed values of $\gamma$ and $c_0$, we do not observe a clear trend in how the convergence time changes as we vary $\delta$.
One commonality between the SBM graphs, the ER graphs, and the complete graph is that $\delta = 1$ gives the fastest convergence times.
For a wide range of fixed values of $\gamma$ and $\delta$, we also observe that the convergence time tends to decreases as we increase $c_0$ for both our adaptive-confidence HK model and the baseline HK model.

%-----------------------------
\subsection{Number of major clusters in simulations on the Reed College network}\label{appendix:reed_clusters}
In our simulations of our adaptive-confidence HK model on the {\sc Facebook100} networks, the Reed College network (see \cref{fig:HK_reed_clusters}) has 
more major opinion clusters than the other universities for very small initial confidence bounds $c_0$ (specifically, $c_0 \in \{0.02, 0.03, 0.04\}$). This difference 
may arise from the small size of the Reed College network in concert with our definition of major cluster.
For example, a final opinion cluster with 20 nodes is a major cluster for the Reed College network (which has 962 nodes in its LCC), but an opinion cluster of that size is a minor cluster for 
the UC Santa Barbara network (which has 14,917 nodes in its LCC).

\begin{figure} [htb]
\centering
\includegraphics[width=0.8\columnwidth]{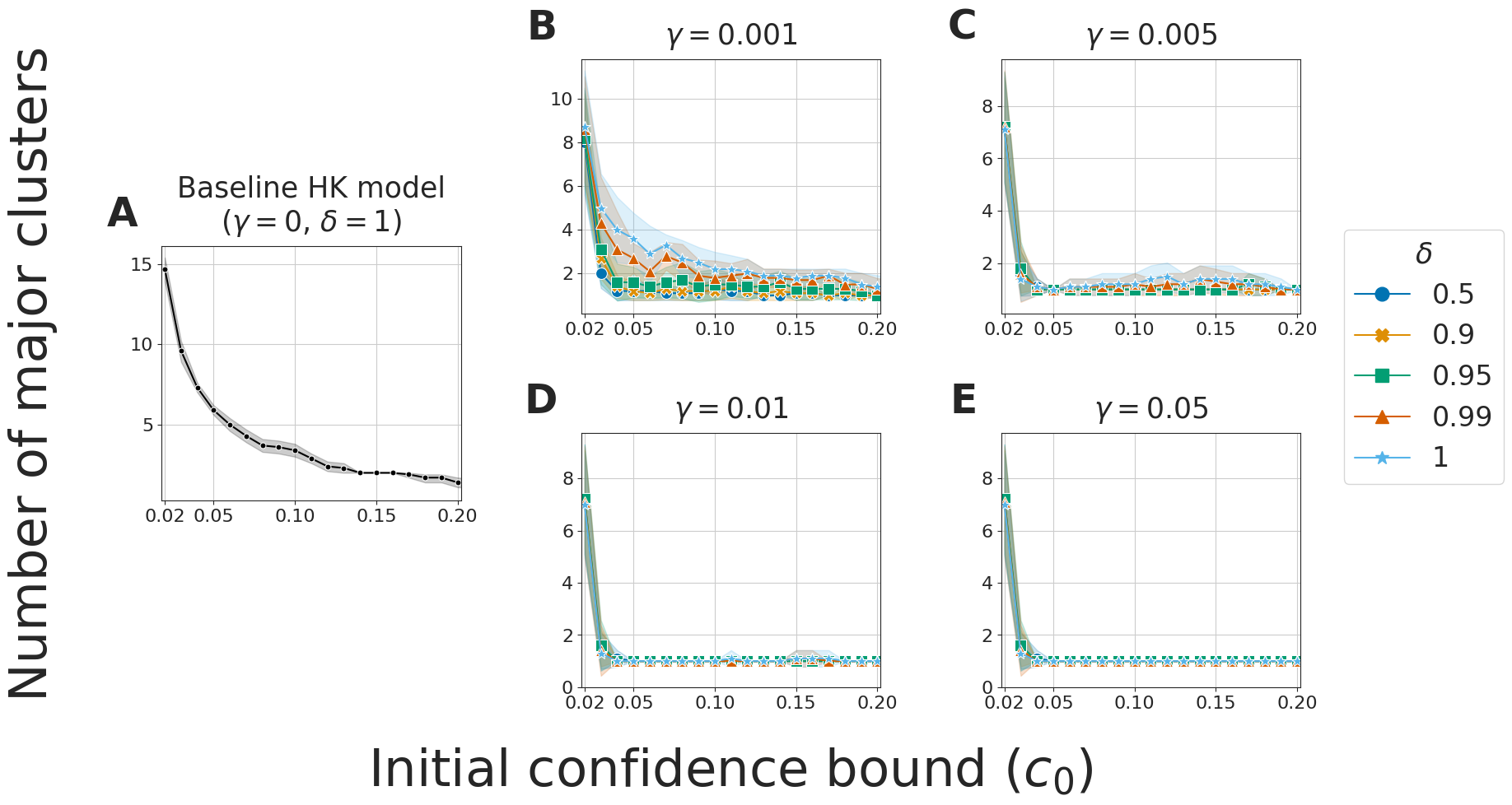}
\caption{The numbers of major clusters in simulations of our adaptive-confidence HK model on the Reed College network for various combinations of the BCM parameters $\gamma$, $\delta$, and $c_0$.}
\label{fig:HK_reed_clusters}
\end{figure}

%%%%%%%%%%%%%%%%%%%%%%%

\section{Additional results and discussion of our numerical simulations of our adaptive-confidence DW model}\label{appendix:DW_simulations}

We now further discuss our simulations of our adaptive-confidence DW model. We simulate this model a 100-node complete graph and the {\sc NetScience} network.
We simulate our adaptive-confidence DW model with the values of the BCM parameters (namely, the confidence-increase parameter $\gamma$, the confidence-decrease parameter $\delta$, the initial confidence bound $c_0$, and the compromise parameter $\mu$) in \cref{tab:parameters}.
See \cref{tab:DW_trends} for a summary of the trends for the two networks.

We explore the dependence of the numbers of major and minor clusters, the Shannon entropy $H(T_f)$ (see equation \cref{eq:entropy}), the weighted-average edge fraction $W(T_f)$ (see equation \cref{eq:weighted_avg}), and the convergence time 
$T_f$ on the initial confidence bound $c_0$. For each value of $(\gamma, \delta)$, we generate one plot; each plot has one curve for each value of the compromise parameter $\mu$.
Each point in our plots is the mean of $10$ numerical simulations (from 10 sets of initial opinions) with one BCM parameter set $(\gamma, \delta, c_0, \mu)$. We also show one standard deviation from the mean.
All plots, including those that we do not present in this appendix, are available in our \href{https://gitlab.com/graceli1/Adaptive-Confidence-BCM}{code repository}.

% ----------------------------------------

\subsection{A complete graph}\label{sec:DW_complete}
We first discuss our simulations of our adaptive-confidence DW model on a 100-node complete graph. 
In the present section, we show plots of the numbers of major opinion clusters (see \cref{fig:DW_complete_clusters}) and the weighted-average edge fractions $W(T_f)$ (see \cref{fig:DW_complete_WT}).

\begin{table}[htb]
\small
\centering
\caption{\label{tab:DW_bailout} Summary of the numbers of simulations of our adaptive-confidence DW model that reach the bailout time of $10^6$ time steps. For each combination of the BCM parameters ($\gamma$, $\delta$, $c_0$, and $\mu$), we run 10 simulations, which each have a different set of initial opinions. 
In each table entry, the focal number is the number of simulations that reach the bailout time and the number in parentheses is the number of those simulations for which we are {also} unable determine the final opinion clusters. We run our simulations with $(\gamma, \delta) = (0.1, 0.5)$ to convergence (i.e., without a bailout time); for those simulations, we do not track the number of opinion clusters at the bailout time.}
\begin{tabular}{cc|cccccc}
\hline\hline
\multicolumn{1}{l}{} & \multicolumn{1}{l|}{} & \multicolumn{6}{c}{\begin{tabular}[c]{@{}c@{}}Number of simulations that reach bailout\\ (number of simulations for which we are also\\unable to determine the final opinion clusters)\end{tabular}} \\ \cline{3-8} 
 &  & \multicolumn{3}{c|}{$\mu=0.1$} & \multicolumn{2}{c|}{$\mu=0.3$} & $\mu=0.5$ \\ \cline{3-8} 
 &  & $c_0=0.1$ & $c_0=0.2$ & \multicolumn{1}{c|}{$c_0=0.3$} & $c_0=0.1$ & \multicolumn{1}{c|}{$c_0=0.2$} & $c_0=0.1$ \\ \hline
\multirow{3}{*}{$\gamma=0.1$} & \multicolumn{1}{c|}{$\delta=0.3$} & 9 (7) & 2 (2) & \multicolumn{1}{c|}{1 (1)} & 0 & \multicolumn{1}{c|}{0} & 0 \\
 & \multicolumn{1}{c|}{$\delta=0.5$} & 8 & 1 & \multicolumn{1}{c|}{0} & 1 & \multicolumn{1}{c|}{0} & 0 \\
 & \multicolumn{1}{c|}{$\delta=0.7$} & 9 (6) & 2 (2) & \multicolumn{1}{c|}{1 (1)} & 2 (0) & \multicolumn{1}{c|}{0} & 0 \\ \hline
\multirow{3}{*}{$\gamma=0.3$} & \multicolumn{1}{c|}{$\delta=0.3$} & 9 (5) & 0 & \multicolumn{1}{c|}{0} & 2 (1) & \multicolumn{1}{c|}{0} & 2 (0) \\
 & \multicolumn{1}{c|}{$\delta=0.5$} & 8 (7) & 0 & \multicolumn{1}{c|}{0} & 2 (2) & \multicolumn{1}{c|}{0} & 0 \\
 & \multicolumn{1}{c|}{$\delta=0.7$} & 7 (4) & 0 & \multicolumn{1}{c|}{0} & 5 (3) & \multicolumn{1}{c|}{2 (1)} & 0 \\ \hline
\multirow{3}{*}{$\gamma=0.5$} & \multicolumn{1}{c|}{$\delta=0.3$} & 9 (6) & 0 & \multicolumn{1}{c|}{0} & 2 (2) & \multicolumn{1}{c|}{1 (0)} & 0 \\
 & \multicolumn{1}{c|}{$\delta=0.5$} & 8 (4) & 0 & \multicolumn{1}{c|}{0} & 2 (1) & \multicolumn{1}{c|}{0} & 0 \\
 & \multicolumn{1}{c|}{$\delta=0.7$} & 6 (4) & 0 & \multicolumn{1}{c|}{0} & 7 (4) & \multicolumn{1}{c|}{0} & 1 (1) \\ \hline\hline
\end{tabular}
\end{table}

Our adaptive-confidence DW model tends to converge more slowly than both the baseline DW model and our adaptive-confidence HK model.
Our simulations of our adaptive DW model often reach the bailout time, particularly for small values of $c_0$ and $\mu$.
In \cref{tab:DW_bailout}, we indicate the numbers of simulations that reach the bailout time.
In some simulations, despite reaching the bailout time, we are still able to identify the final opinion clusters. However, the maximum difference in the opinions of the nodes in these clusters is not within our tolerance value (see \cref{eq:stopping_criterion}) of $0.02$ for our adaptive-confidence DW model.
In those instances, we still use the cluster information to calculate the numbers of major and minor opinion clusters, 
the Shannon entropy $H(T_f)$,
and the weighted-average edge fraction $W(T_f)$. 
For our simulations of our adaptive-confidence DW model with $(\gamma, \delta) = (0.1, 0.5)$, we run each simulation to convergence (i.e., until we reach the stopping condition that we described in \cref{sec:simulations_specs}). We plot the results of these simulations in
\cref{fig:DW_complete_clusters}E and \cref{fig:DW_complete_WT}B.
Although some simulations reach the bailout time, the information that we are able to obtain about the opinion clusters  
(from both the simulations that we run to convergence and the simulations that reach the bailout time) give us confidence in the 
trends in \cref{tab:DW_trends}.

\begin{figure} [ht]
\centering
\includegraphics[width=0.9\columnwidth]{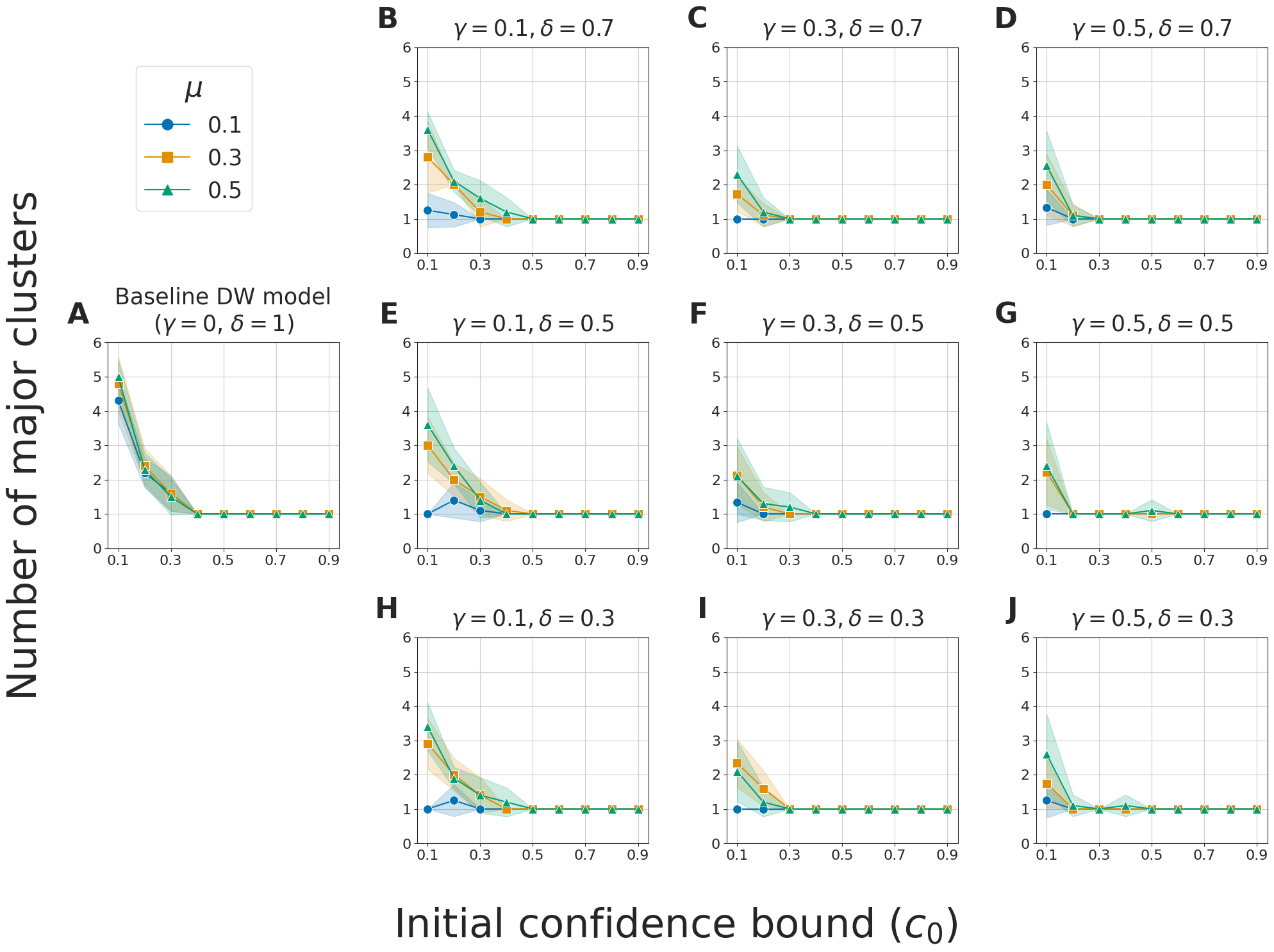}
\caption{The numbers of major clusters in simulations of (A) the baseline DW model and (B--J) our adaptive-confidence DW model on a 100-node complete graph for various combinations of the BCM parameters $\gamma$, $\delta$, $c_0$, and $\mu$.
In this figure and subsequent figures, we do not use simulations in which we are unable to determine the final opinion clusters (see \cref{tab:DW_bailout}) to calculate the means and standard deviations.
In (E), in which we show our simulations with $(\gamma, \delta) = (0.1, 0.5)$, we run all of our simulations to convergence (i.e., we ignore the bailout time) and use all of our simulations to calculate the mean numbers of major opinion clusters.
}
\label{fig:DW_complete_clusters}
\end{figure}

In \cref{fig:DW_complete_clusters}, we observe for a wide range of BCM parameter values that our adaptive-confidence DW model yields fewer major clusters (i.e., it encourages more consensus) than the baseline DW model. When $c_0 \geq 0.5$, our adaptive-confidence DW model and the baseline DW model always reach consensus.
For fixed values of $\gamma$, $\delta$, and $c_0$, when our adaptive-confidence DW model does not reach consensus, 
decreasing the compromise parameter $\mu$ tends to result in fewer major clusters.
By contrast, $\mu$ has little effect on the number of major clusters in the baseline DW model.   
Increasing $\gamma$ with the other BCM parameters (i.e., $\delta$, $c_0$, and $\mu$) fixed also tends to result in fewer major clusters.
Changing $\delta$ with the other parameters fixed has little effect on the number of major clusters. In fact,
changing $\delta$ with the other parameters fixed appears to have little effect on any of the computed quantities, so we show results only for $\delta = 0.5$ in our subsequent figures.
In our \href{https://gitlab.com/graceli1/Adaptive-Confidence-BCM}{code repository}, we include plots for the other examined values of $\delta$.

We observe very few minor clusters in our simulations of our adaptive-confidence DW model on the 100-node complete graph. For each BCM parameter set $(\gamma, \delta, c_0, \mu)$, the mean number of minor clusters in our 10 simulations is bounded above by $1$.
Consequently, the number of major clusters and Shannon entropy follow similar trends.
Overall, in our simulations on the 100-node complete graph, our adaptive-confidence DW model encourages more consensus than the baseline DW model and this difference between these two models becomes more pronounced for larger values of the confidence-increase parameter $\gamma$ and smaller values of the compromise parameter $\mu$. 

In \cref{fig:DW_complete_WT}, we show the weighted-average edge fraction $W(T_f)$ (see \cref{eq:weighted_avg}).
The baseline DW model always has $W(T_f) = 1$.
By contrast, for sufficiently small initial confidence values $c_0$, our adaptive-confidence DW model yields $W(T_f) < 1$. 
For $\mu = 0.1$ and small $c_0$ (specifically, $c_0 \leq 0.3$), our adaptive-confidence DW model can reach consensus with $W(T_f) < 1$. 
As in our adaptive-confidence HK model (see our discussion in \cref{sec:HK_complete}), this observation indicates that some adjacent nodes in the same final opinion cluster are not receptive to each other.

\begin{figure} [ht]
\centering
\includegraphics[width=0.9\columnwidth]{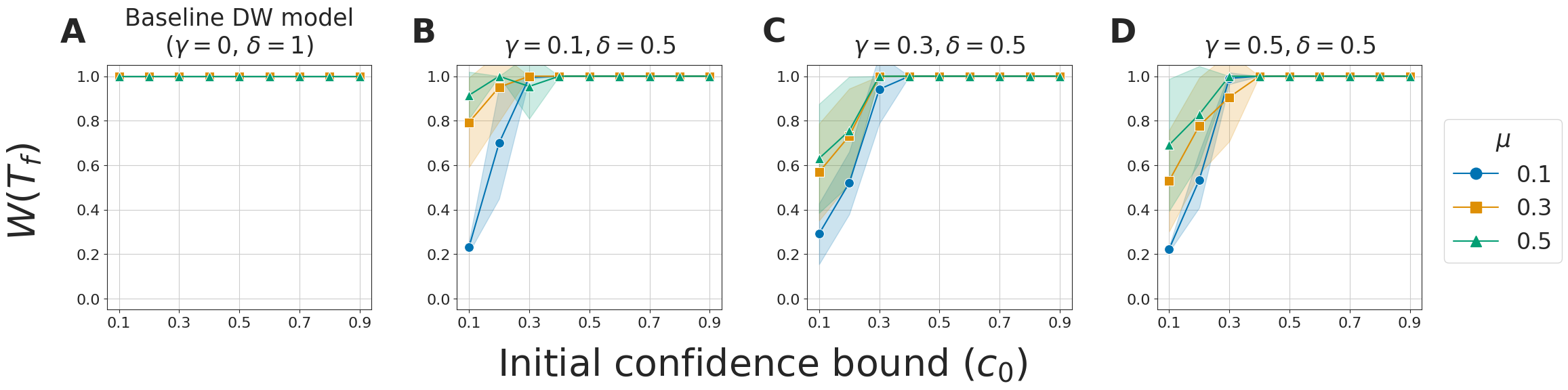}
\caption{
The weighted-average edge fraction $W(T_f)$ (see equation \cref{eq:weighted_avg})
in simulations of (A) the baseline DW model and (B--D) our adaptive-confidence DW model on a 100-node complete graph for various combinations of the BCM parameters $\gamma$, $\delta$, $c_0$, and $\mu$.
In (E), in which we show our simulations with $(\gamma, \delta) = (0.1, 0.5)$, we run all of our simulations to convergence (i.e., we ignore the bailout time) and use the resulting final opinion clusters.
}
\label{fig:DW_complete_WT}
\end{figure}

For fixed values of $c_0 \leq 0.3$ and $\mu$, our adaptive-confidence DW model tends to converge more slowly than the baseline DW model.
Additionally, when we fix the other BCM parameters (i.e., $\gamma$, $\delta$, and $\mu$), the convergence time tends to increase as we decrease $c_0$.
As we showed in \cref{tab:DW_bailout}, for small values of $c_0$ (specifically, $c_0 \in\{0.1, 0.2\}$), more simulations reach the bailout time as we decrease $\mu$.
In both our adaptive-confidence DW model and the baseline DW model, $\mu = 0.1$ yields longer convergence times than $\mu \in \{0.3, 0.5\}$ for fixed values of $\gamma$, $\delta$, and $c_0$.

% ----------------------------------------
\subsection{Network of network-scientist coauthorships}\label{sec:DW_netscience}
We now discuss our simulations of our adaptive-confidence DW model on the {\sc NetScience} network \cite{netscience}, which is a network of network scientists with unweighted and undirected edges that encode paper coauthorships.

For the {\sc NetScience} network and fixed values of $c_0$ and $\mu$, our adaptive-confidence DW model tends to have at least as many major opinion clusters (see \cref{fig:DW_netscience_major}) and minor opinion clusters (see \cref{fig:DW_netscience_minor}) as the baseline DW model. 
In \cref{fig:DW_netscience_minor}, we see for $c_0 \leq 0.5$ that both our adaptive-confidence DW model and the baseline DW model yield many more minor clusters for the {\sc NetScience} network than for the 100-node complete graph.
For values of $c_0$ that are near the transition between consensus and opinion fragmentation (specifically, $c_0 \in \{0.3, 0.4, 0.5\}$), our adaptive-confidence DW model yields noticeably more major clusters and minor clusters than the baseline DW model. The transition between consensus and fragmentation appears to occur for a larger threshold in our adaptive-confidence DW model than in the baseline DW model.
For the {\sc NetScience} network (and unlike for the $100$-node complete graph), changing the value of $\mu$ with the other BCM parameters fixed appears to have little effect on the numbers of major and minor opinion clusters.

For the {\sc NetScience} network and fixed values of $c_0$ and $\mu$, 
our adaptive-confidence DW model has convergence times that are similar to those of the baseline DW model.
All of our simulations of our adaptive-confidence DW model on the {\sc NetScience} network converge before reaching the bailout time.
We obtain the longest convergence times for $c_0 = 0.3$. By contrast, for the 100-node complete graph, the convergence time increases as we decrease $c_0$ and many simulations reach the bailout time for $c_0 \in \{0.1, 0.2\}$.
In both our adaptive-confidence DW model and the baseline DW model, $\mu = 0.1$ yields longer convergence times than $\mu \in \{0.3, 0.5\}$ for fixed values of $\gamma$, $\delta$, and $c_0$.
We do not observe a clear trend in how the convergence time changes either as a function of $\gamma$ (with fixed $\delta$, $c_0$, and $\mu$) or as a function of $\delta$ (with fixed $\gamma$, $c_0$, and $\mu$).

\begin{figure} [ht]
\centering
\includegraphics[width=0.8\columnwidth]{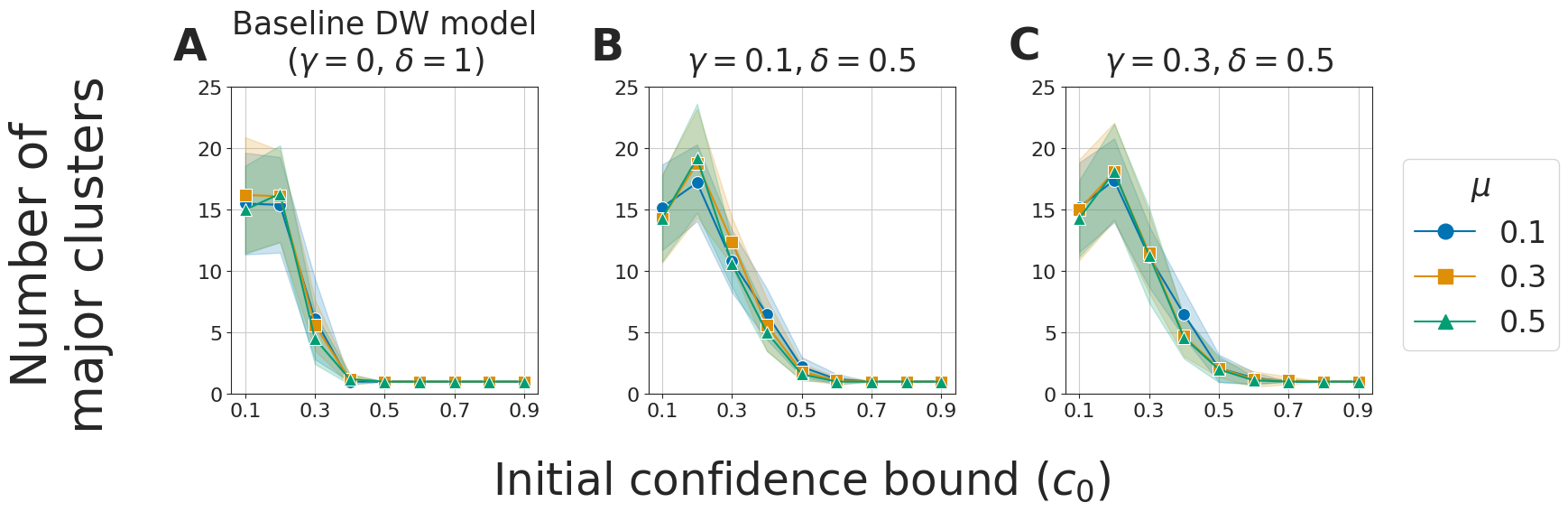}
\caption{The numbers of major clusters in simulations of (A) the baseline DW model and (B, C) our adaptive-confidence DW model on the {\sc NetScience} network for various combinations of the BCM parameters $\gamma$, $\delta$, $c_0$, and $\mu$.}
\label{fig:DW_netscience_major}
\end{figure}

\begin{figure} [ht]
\centering
\includegraphics[width=0.8\columnwidth]{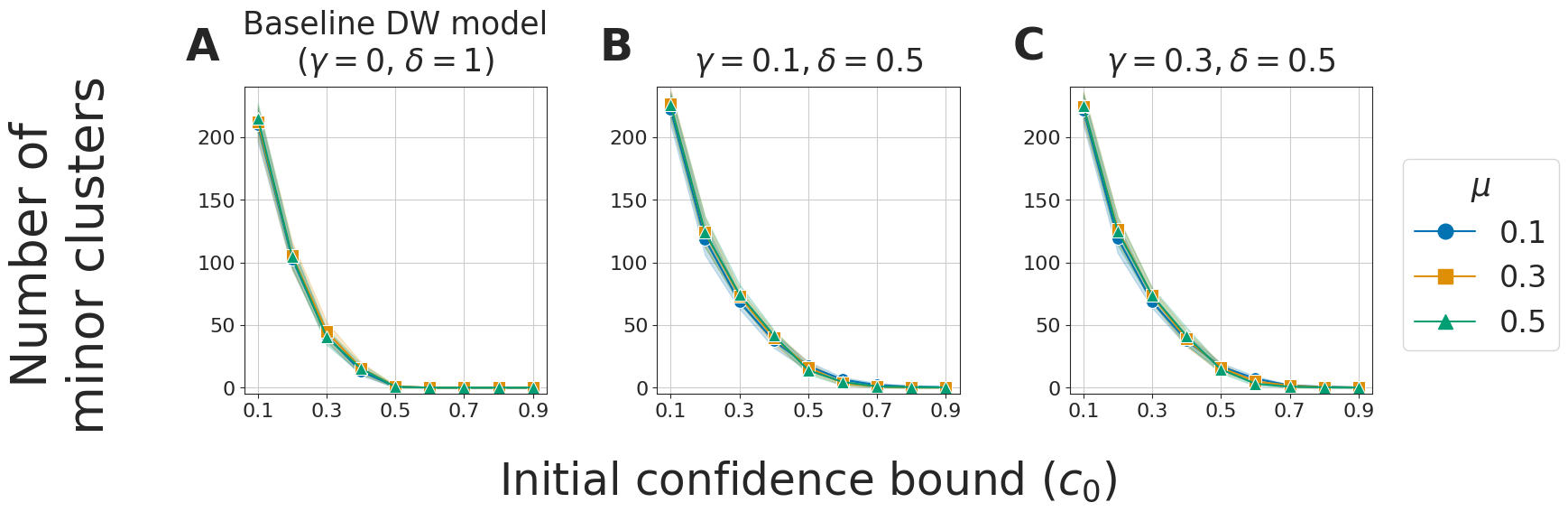}
\caption{The numbers of minor clusters in simulations of (A) the baseline DW model and (B, C) our adaptive-confidence DW model on the {\sc NetScience} network for various combinations of the BCM parameters $\gamma$, $\delta$, $c_0$, and $\mu$.}
\label{fig:DW_netscience_minor}
\end{figure}

%%%

\section*{Acknowledgements}
We thank Weiqi Chu, Gillian Grindstaff, Abigail Hickok, and the participants of UCLA’s Networks Journal Club for helpful discussions and comments. 
Additionally, we thank Weiqi Chu for her ideas and helpful discussions for the proof of \cref{thm:effgraph_baseline_DW}.
We also thank two anonymous referees for their helpful comments, and we thank Carmela Bernardo, Damon Centola, PJ Lamberson, and Jim Moody for pointers to helpful references. 

%%%%

%\bibliographystyle{siamplain}
%\bibliography{references-v10}

%%%%%

%%%%%

\end{document}